%% file: main.tex
\documentclass{article}
\usepackage[utf8]{inputenc}
\usepackage{amsmath,amsfonts,amsthm}
\usepackage{natbib}
\usepackage{booktabs}
\usepackage{enumerate}
\usepackage{graphicx}
\usepackage{longtable}
\usepackage{setspace}
\usepackage[section]{placeins}
\usepackage[T1]{fontenc}
\usepackage{svg}
\usepackage{adjustbox}
\usepackage{svg-extract}

\setcitestyle{authoryear}
\defcitealias{sffa}{\textit{SFFA v. Harvard,} 2022}

\newcommand{\E}{\mathbb{E}}

\theoremstyle{plain}
\newtheorem{proposition}{Proposition}
\theoremstyle{definition}
\newtheorem{definition}[proposition]{Definition}
\newtheorem{assumption}[proposition]{Assumption}
\newtheorem{example}[proposition]{Example}
\theoremstyle{remark}
\newtheorem*{remark}{Remark}
\usepackage{hyperref}
\usepackage{import}

\title{Quantifying the Reliance of Black-Box Decision-Makers on Variables of Interest}
\author{Daniel Vebman \\ Cornell University \\ {\small dov3@cornell.edu}}
\date{May 2023}

\begin{document}

\onehalfspacing

\maketitle

\section*{Abstract}

This paper introduces a framework for measuring how much black-box decision-makers rely on variables of interest. The framework adapts a permutation-based measure of variable importance from the explainable machine learning literature. With an emphasis on applicability, I present some of the framework's theoretical and computational properties, explain how reliance computations have policy implications, and work through an illustrative example. In the empirical application to interruptions by Supreme Court Justices during oral argument, I find that the effect of gender is more muted compared to the existing literature's estimate; I then use this paper's framework to compare Justices' reliance on gender and alignment to their reliance on experience, which are incomparable using regression coefficients.

\hfill \break Keywords: econometrics, interpretability, fairness, explainable machine learning, gender.

\section{Introduction}

Decision-makers routinely choose among some menu of options without explaining why.
We are interested in understanding which observed variables are important to the decision-maker: does the judge rely on sex, being more lenient towards women \citep{Goulette}? Does the doctor rely on race, assuming that Black people have higher pain tolerance \citep{Staton}? Or, as the Supreme Court recently decided, do (Harvard) admissions officers discriminate against Asians \citepalias{sffa}? 
Understanding a decision-making process has inherent value and direct policy implications. To correct an iniquity, do judges need to be trained about their cognitive biases, do medical textbooks need to be rewritten, or do admissions practices need to change? 

The difficulty is that the decision-maker is a black box, whom we cannot query for new data points.
Human brains are figuratively black boxes and literally neural nets.
Combining the CS literature on ML interpretability with econometric work on the prediction of counterfactual choices---including in partially identified settings---presents a novel way to understand and, if needed, adjust how decision-makers make choices in the real world.

This paper introduces a framework for measuring how much black-box decision-makers rely on variables of interest. The approach is inspired by a permutation-based measure of variable importance from the explainable machine learning literature; \citet{Breiman_Manual} originally presented this method in his study of random forests, and \citet*{Fisher} generalized it to arbitrary models. This paper's contributions are:
\begin{enumerate}
    \item I generalize Fisher et al.'s approach beyond ML, to human behavior. 
    \item I prove in proposition \ref{prop:square-loss-r-vs-b} that this framework encapsulates and can test for conditional statistical parity, a fairness metric from the explainable machine learning literature.
    \item Propositions \ref{prop:square-loss-r-vs-b} and \ref{prop:equivalent-ranking} theoretically justify Fisher et al.'s normalization by a baseline.
    \item Studying interruptions by Supreme Court Justices in section \ref{sec:supreme-court-interruptions}, I find smaller effects than the existing literature. I also rank each Justice's reliance on gender, alignment, and experience, which were previously not compared.
\end{enumerate}

\section{Related Work}

\subsection{Variable Importance Techniques from Economics and the Social Sciences}

Importance analysis quantifies how much certain variables contribute to uncertainty around the outcome \citep{Coyle}. Correlation- or variance-based measures are the most common importance measures. Simple correlation coefficients and contributions to variance fall in this category. 
Probability-, elasticity-, and information/entropy-based measures are also common in the economics and social sciences literatures \citep{Coyle}. 

Variable importance is distinct from and offers advantages over typical alternatives.
First, consider a linear model. Different units prevent direct comparisons of coefficients to each other. For example, as explored in the empirical application in section \ref{sec:supreme-court-interruptions}, consider a model for the number of times a Supreme Court Justice interrupts an advocate during oral argument, $Interruption Rate = \beta_0 + \beta_1 Gender + \beta_2 Experience + \beta_3 Justice And Advocate Agree$. We can directly compare $\beta_1$ to $\beta_3$ since the covariates they multiply are both binary. It is nonsensical, however, to compare a marginal increase in $Experience$ to a change in $Gender$ from male to female -- the units are completely different. The variable-importance framework presented here overcomes this limitation and allows direct comparisons of any two variables. 

A second advantage of variable importance is that it is sensitive to the relationship between the different inputs. As illustrated in the school admissions example in section \ref{sec:school-admissions}, the coefficients on the model do not alone dictate how much a variable affects the model's output. More generally, the distinction between statistical and economic significance suggests an important analogue to variable importance for decision-makers: even though a variable might have a large effect on a hypothetical input to the model, it might only have limited importance for a typical input, in practice. 

\subsection{Variable Importance Techniques from Explainable Machine Learning}

The ML explainability literature studies \textit{post-hoc} methods to understand why a model arrives at its outputs \citep{Barredo}.
It includes the general class of variable importance (or feature relevance) measures, which quantify how much a prediction model's accuracy depends on the information in each covariate. This class contains Shapley values \citep[e.g.,][]{Lundberg}, perturbation-based measures \citep[e.g.,][]{Robnik}, and permutation-based measures \citep[e.g.,][]{Breiman}. Perturbation-based measures observe the change in a loss function for a small or infinitesimal change to the inputs. Permutation-based measures observe the change in loss after shuffling independent observations together to sever the correlation between the variable of interest and the other covariates/label.

Permutation-based measures are well-suited to black-box models because they only require a mapping from inputs to outputs. They do not need to rely on gradients or any other knowledge of the model.
Indeed, model-agnostic explanations are increasingly popular because they are widely applicable \citep{Balagopalan}.
Proposed by \citet{Ribeiro}, Local Interpretable Model-agnostic Explanations (LIME) and its variations are among the most well-known such approaches. 
LIME is a local technique, which explains how the model made a single decision by approximating a small region of the decision boundary. 
In contrast, global methods approximate the entire function with a more interpretable surrogate (like a tree-based or sparse linear model) or otherwise summarize the entire decision model. 

This paper centers on a permutation-based measure of variable importance. \citet{Breiman, Breiman_Manual} first developed this approach in his study of random forests. \citet{Gregorutti_Michel_Saint-Pierre_2017} use this method in a variable-selection algorithm, specifically for random forests that model the conditional expectation function. \citet{Fisher} suggest this permutation-based approach as a generic measure of variable importance for any model and with arbitrary loss function; they further study how to bound this measure for sets of models that perform roughly equally well on the same prediction problem.

Like Gregorutti et al., I study the conditional expectation function only; however, whereas they study random forests and typically use the square loss, I generalize to any estimator and arbitrary loss. 
I use almost the same generic definition of reliance as Fisher et al.; but, unlike them, I study the conditional expectation function in the context of human decision-making and discuss how reliance values could be used in various public-policy settings. 

\section{Framework}

\subsection{Setup}

A decision-maker chooses an alternative $Y \in \mathcal{Y}$ depending on covariates, which we partition into $X_1 \in \mathcal{X}_1$ and $X_2 \in \mathcal{X}_2$. 
In this part, we define a formal measure of the decision-maker's reliance on $X_1$. 

We define reliance by rephrasing our question as a prediction problem:
the decision rule induces a joint distribution $(Y, X_1, X_2) \sim P$ on the choice and covariates. 
We have an oracle model $\mathbb{E}[Y | X_1, X_2]$, where $(Y, X_1, X_2) \sim P$, which tells us the decision-maker's choice. 
The oracle allows us to make counterfactual queries to the choice function. 

By definition, the oracle assumes that its inputs and outputs are distributed according to $P$. To measure the oracle's reliance on $X_1$, we therefore observe how much the oracle errs when we replace $X_1$ with noise.
Specifically, we make $X_1$ completely uninformative of $Y$ and $X_2$, while preserving their marginal distributions. We make this intuitive notion precise in the next section.

\subsection{Definition of Reliance}

Let $(Y^a, X_1^a, X_2^a)$ and $(Y^b, X_1^b, X_2^b)$ be two independent draws from $P$.
Splice $X_1^b$ into the $a$ draw to create the coupling $(Y^a, X_1^b, X_2^a)$.
$X_1^b$ and $(Y^a, X_2^a)$ have the same marginal distributions as before, but they are now independent.
We want to compare $Y^a$ to the oracle's prediction for the pair $(X_1^b, X_2^a)$.

\begin{definition}[Oracle's prediction]
    \label{def:oracle-prediction}
    The oracle's prediction for the pair $(x_1, x_2) \in \mathcal{X}_1 \times \mathcal{X}_2$ is
    $$f(x_1, x_2) = \E[Y \mid X_1 = x_1, X_2 = x_2],$$ 
    where $(Y, X_1, X_2) \sim P$ as before. 
\end{definition}

We want to measure how much the oracle errs when we feed it an $X_1$ that is completely uninformative of $Y$. To quantify this change, we need a loss function $L$. We also require the following technical assumption:

\begin{assumption}
    \label{assump:square-loss-technical-assumption-abs-cont}
    Assume that the coupling $P_{X_1^b, X_2^a}$ is absolutely continuous with respect to $P_{X_1^a, X_2^a}$.
    Recall $P_{X_1^b, X_2^a} = P_{X_1}P_{X_2}$ and $P_{X_1^a, X_2^a} = P_{X_1, X_2}$.
    The subscripts on $P$ refer to the respective marginal distributions.
\end{assumption}

\begin{remark}
    This assumption is necessary so that the oracle's prediction is well-defined over the shuffling. For example, if $X_1$ and $X_2$ are binary but are never equal, then this assumption fails. When we try to shuffle $X_1$, we will fail to compute $\E[Y \mid X_1 = 0, X_2 = 0]$ since the conditioned event has probability 0.
\end{remark}

We can now formally define model reliance:

\begin{definition}[Reliance on $X_1$]
    \label{def:reliance}
    Given a loss function $L : \mathcal{Y} \times \mathcal{Y} \rightarrow \mathbb{R}$, and a partition $(X_1, X_2)$ of the covariates, the reliance on $X_1$ is
    $$r = \E_{Y^a, X_1^b, X_2^a} L(Y^a,\, f(X_1^b, X_2^a))$$
    where $(Y^a, X_2^a) \sim P_{Y,X_2}$ and $X_1^b \sim P_{X_1}$ are independent.
\end{definition}

In general, we might require that the loss function $L$ admit some or all of the following kinds of statements:
\begin{enumerate}
    \item Rankings of variables within one distribution:
    The Justice relies more on gender than on experience.
    \item Rankings of one variable across distributions:
    Justice P relies more on gender than Justice Q does.
    \item Rankings of different variables across distributions:
    Justice P relies more on gender than Justice Q relies on experience.
\end{enumerate}
The theoretical exposition will justify that these three statements are sensical, and the applications will demonstrate that they are natural and valuable. We will make all these kinds of statements in the application to Supreme Court Justices in section \ref{sec:supreme-court-interruptions}.
First, I provide a few examples of loss functions and explain their properties.

\begin{example}[Square Loss]
    \label{ex:square-loss}
    A simple choice for $L$ is the square loss, 
    $$L(y, \hat{y}) = (y - \hat{y})^2.$$
    That is,
    $$r = \E_{Y^a, X_1^b, X_2^a} (Y^a - f(X_1^b, X_2^a))^2.$$
    We will show in the next proposition that, by using the square loss, we can interpret reliance values with respect to a baseline. Further, the reliance and the baseline are equal if and only if the decision $Y$ is conditionally mean independent of $X_1$ given $X_2$. Thus, this reliance measure can test \textit{conditional statistical parity}, a fairness metric from the machine learning literature.

    \begin{assumption}
        \label{assump:square-loss-technical-assumption-L2}
        Let $(\mathcal{Y} \times \mathcal{X}_1 \times \mathcal{X}_2, \mathcal{F}, P)$ be a probability space.
        In what follows, assume the random vector $(Y, X_1, X_2)$ on this space is $L^2$. 
        As before, $(Y, X_1, X_2)$, $(Y^a, X_1^a, X_2^a)$ and $(Y^b, X_1^b, X_2^b)$ are independent and identically distributed.
    \end{assumption}
    
    \begin{proposition}
        \label{prop:square-loss-r-vs-b}
        To simplify notation, let $f(x_2) = \E[Y \mid X_2 = x_2]$, where $(Y, X_2) \sim P_{Y,X_2}$.
        Define the \textit{baseline reliance}, 
        $$b = \E_{Y^a, X_1^a, X_2^a} (Y^a - f(X_1^a, X_2^a))^2,$$
        where we do not shuffle $X_1$. Then, $r \geq b$. 
        Furthermore, the following are equivalent:
        \begin{enumerate}
            \item $r = b$;
            \item $f(X_1^a, X_2^a) = f(X_1^b, X_2^a)$ a.s.
            with respect to $(X_1^a, X_2^a, X_1^b) \sim P_{X_1,X_2} P_{X_1}$;
            \item $f(X_1^a, X_2^a) = f(X_2^a)$ a.s.
            with respect to $(X_1^a, X_2^a) \sim P_{X_1,X_2}$; and
            \item $Y$ is conditionally mean independent of $X_1$ given $X_2$ almost surely. 
        \end{enumerate}
    \end{proposition}

    \begin{proof}\hyperlink{pf:square-loss-r-vs-b}{Proof in appendix.}\end{proof}

    \begin{remark}
        We prove the chain 1 $\Leftrightarrow$ 2 $\Leftrightarrow$ 3, and 4 simply rephrases 3.
        2 $\Leftrightarrow$ 3 is not trivial because the conditions hold almost surely over related but different distributions.
    \end{remark}

    Because $r \geq b$ always, we justifiably call $b$ the \textit{baseline reliance} of $Y$ on $X_1$. We can interpret $b$ as the intrinsic noise of $P$ --- it is the loss that even the best predictor incurs. 
    \citet{Fisher} broadly define their reliance measure as the ratio $r / b$, but for arbitrary loss function. 
    However, it is generally false that $b$ minimizes the loss for arbitrary $L$, which challenges $b$'s use as a baseline;
    furthermore, as demonstrated in section \ref{sec:comparing-reliance}, this definition unnecessarily complicates comparisons within the same distribution.
    We return to the idea of partialling out the intrinsic noise in example \ref{ex:cross-entropy-loss} and in section \ref{sec:comparing-reliance} on comparing reliance across multiple choice distributions.

    This proposition also provides strong intuition for how to interpret reliance values:
    in the context of binary decisions, conditional mean independence is equivalent to \textit{conditional statistical parity}, a fairness metric from the explainable machine learning literature that captures ``fairness through blindness'' \citep[see][]{DBLP:journals/corr/Corbett-DaviesP17}. In the fairness setting, $X_1$ are the \textit{sensitive} attributes and $X_2$ are the \textit{legitimate} attributes. 
    \citet[][Section 10.2]{Fisher} only identify that reliance is ``related'' to this fairness metric.
    Contributing to the explainable machine learning literature, this proposition proves that $r = b$ is equivalent to and therefore can test for conditional statistical parity.
\end{example}

\begin{example}[Cross-Entropy Loss]
    \label{ex:cross-entropy-loss}
    In a typical binary classification setting where the oracle returns probabilities in the range $[0,1]$, it is unclear how to interpret the actual values taken by the square loss. 
    The cross-entropy loss is commonly used in the machine learning literature in such settings:
    \begin{align*}
        \mathbb{E}_{Y^a, X_1^b, X_2^a} [
            Y^a &\log \E[Y \mid X_1 = X_1^b, X_2 = X_2^a] \\
            &+ (1 - Y^a) \log (1 - \E[Y \mid X_1 = X_1^b, X_2 = X_2^a]) ] 
    \end{align*}
    This equals the expected cross-entropy $H( Y \mid (X_2 = X_2^a) ,\; Y \mid (X_1 = X_1^b, X_2 = X_2^a) )$. A neat interpretation follows from the relation to the Kullback-Leibler divergence:
    \begin{align*}
        &H\left( 
            Y \mid (X_2 = X_2^a), \;
            Y \mid (X_1 = X_1^b, X_2 = X_2^a)
        \right) \\
        &= H(Y \mid (X_2 = X_2^a)) \\
        &\qquad+ D_{KL}\left( 
            Y \mid (X_2 = X_2^a)
            \mid\mid
            Y \mid (X_1 = X_1^b, X_2 = X_2^a)
        \right) \\
        &\text{(cross-entropy = entropy + divergence)}
    \end{align*}
    The Kraft-McMillan theorem establishes the optimal number of bits to code messages that follow a specific distribution.
    Returning to our setting, this theorem implies the following interpretations:
    \begin{enumerate}
        \item The \textit{entropy} term equals the optimal expected number of bits needed to code a draw from $Y \mid (X_2 = X_2^a)$. 
        \item The \textit{cross-entropy} term equals the optimal expected number of bits needed to code a draw from $Y \mid (X_2 = X_2^a)$ if we wrongly assume that our draws come from $Y \mid (X_1 = X_1^b, X_2 = X_2^a)$. 
        \item The \textit{divergence} term equals the excess number of bits we need to code a draw from $Y \mid (X_2 = X_2^a)$ if we mistakenly assume that it is drawn from $Y \mid (X_1 = X_1^b, X_2 = X_2^a)$.
    \end{enumerate}
    It a mistake to assume that the data are drawn from $Y \mid (X_1 = X_1^b, X_2 = X_2^a)$ because $X_1^b \perp X_2^a$. Reliance measures the cost of this mistake.
\end{example}

\begin{example}[Context-Specific Loss]
    However, for public-policy questions, it may be difficult to translate such an information-based interpretation into the language of policy. More broadly, how does one \textit{evaluate} the loss once it is computed?
    For example, consider the scenario where judges decide if to detain ($y = 1$) or release ($y = 0$) a defendant pretrial based on some covariates $x$.
    Suppose we assert that all judges \textit{should} maximize the utility function 
    $$u(y,x) = - y P(S=0 \mid X=x) - \lambda (1 - y) P(S=1 \mid X=x)$$
    where $S = 1$ indicates the event that the defendant with characteristics $X = x$ would skip trial if released, and $S = 0$ means that they would not skip if released. $\lambda$ is a policy-preference weight.
    According to this $u$, judges should minimize the cost of a mistake.
    Next, adjust our reliance measure to be
    $$r = \E_{Y^a, X_1^a, X_2^a, X_1^b} L(Y^a, \E[Y \mid X_1 = X_1^b, X_2 = X_2^a];\; (X_1^a, X_2^a))$$
    where
    $$L(y, \hat{y}; x) = u(y, x) - u(\hat{y}, x).$$
    This captures how much the judge relies on $X_1$ to attain his solution to the maximum-utility problem.
    If $r > 0$, then the judge uses information in $X_1$ to increase utility, i.e., minimize mistakes.
    If $r < 0$, then the judge's reliance on $X_1$ lowers utility, i.e., causes more mistakes.
\end{example}

\subsection{Comparing Reliance}
\label{sec:comparing-reliance}

As defined, we can already compare a decision-maker's reliance on a variable $X_1$ to their reliance on another variable $X_1'$ because the expected values are taken over the same distribution. However, two decision-makers impose distinct joint distributions over their choices and the covariates.
For example in section \ref{sec:supreme-court-interruptions}, which investigates interruptions by Supreme Court Justices during oral argument, each Justice is a separate decision-maker.

We can compute a normalized measure of reliance across multiple distributions by joining all of the decision-makers' distributions. 
In particular, consider $n$ decision-makers. For each $1 \leq i \leq n$, we have a choice $Y_i$, a partition of the covariates $(X_{1i}, X_{2i})$, and a joint distribution $P_i$ over choices and covariates. Furthermore, let $\mathcal{P}$ denote the (arbitrary) coupling of all the $P_i$'s. In the example of multiple Supreme Court Justices interrupting the same advocate, the $P_i$'s are not independent. 

\begin{definition}[Cross-Distribution Reliance]
    \label{def:cross-distribution-reliance}
    Fix some loss function $\mathcal{L} : (\mathcal{Y}_1 \times \cdots \times \mathcal{Y}_n)^2 \rightarrow \mathbb{R}$, and let $f_i : \mathcal{X}_{1i} \times \mathcal{X}_{2i} \rightarrow \mathcal{Y}_i$ be the prediction function for decision-maker $i$, as in definition \ref{def:oracle-prediction}.
    The cross-distribution reliance of decision-maker $k$ on $X_{1k}$ is
    $$r_k^\times = \E \mathcal{L}\left( ( Y_i^a )_{i=1}^n, \; (\, f_i(s_k(X_{1i}^a, X_{1i}^b), X_{2i}^a) \, )_{i=1}^n \right),$$
    where the shuffle function $s_k(X_{1i}^a, X_{1i}^b)$ equals $X_{1i}^b$ if $i = k$ and $X_{1i}^a$ if $i \ne k$.
    The expectation is over the independent draws $(Y_i^a, X_{1i}^a, X_{2i}^a)_{i=1}^n$ and $(Y_i^b, X_{1i}^b, X_{2i}^b)_{i=1}^n$ from $\mathcal{P}$.
\end{definition}

Observe that the cross-distribution reliance coincides with the original definition of reliance if $n = 1$. In effect, we have stacked all the decision-makers into one super-decision-maker. A special case is when $\mathcal{L}$ is additively separable with respect to $i$:

\begin{proposition}[Equivalent Ranking]
    \label{prop:equivalent-ranking}
    Suppose the cross-distribution loss function $\mathcal{L}$ is additively separable, i.e., 
    $$\mathcal{L}( (y_i)_{i=1}^n, (\hat{y}_i)_{i=1}^n ) = \sum_{i=1}^n L_i(y_i, \hat{y}_i),$$
    where $L_i : \mathcal{Y}_i \times \mathcal{Y}_i \rightarrow \mathbb{R}$ for each $i$.
    For each decision-maker $i$, define the baseline reliance:
    $$b_i = \E_{Y_i^a, X_{1i}^a, X_{2i}^a} L_i\left( Y_i^a,\; f_i(X_{1i}^a, X_{2i}^a) \right),$$
    where the expectation is taken with respect to $(Y_i^a, X_{1i}^a, X_{2i}^a) \sim P_i$ for each $i$, as indicated.
    Then, 
    $$r_j^\times < r_k^\times \Leftrightarrow r_j - b_j < r_k - b_k,$$
    where $r_k$ is the reliance on $X_k$ as in definition \ref{def:reliance}.
    That is, the ranking by $r_k - b_k$ is equivalent to the ranking by $r_k^\times$.
\end{proposition}

\begin{proof}\hyperlink{pf:equivalent-ranking}{Proof in appendix.}\end{proof}

\begin{remark}
    Letting $L_i$ be the cross-entropy as in example \ref{ex:cross-entropy-loss}, the difference $r_k - b_k$ equals the KL divergence. Thus, we normalize the reliance measures by partialling out each distribution's intrinsic noise.
\end{remark}

Proposition \ref{prop:equivalent-ranking} is an important result because $r_k - b_k$ is much simpler to compute than $r_k^\times$. 
Furthermore, the equivalence rigorously justifies Fisher et al.'s normalization by the un-shuffled expected loss.
In section \ref{sec:supreme-court-interruptions}, we apply this proposition to determine which Justices rely the most on gender when interrupting advocates.

\subsection{Acting on Reliance Values}

Beyond providing insight into the decision-making process, reliance values are an actionable metric for a wide range of real-world problems. 

\begin{example}[Enforcing Conditional Statistical Parity]
    Consider for example the admissions decision setting, like in the lawsuit \citetalias{sffa}.
    We might assert that the admissions decision should not directly rely on race ($X_1$). 
    As proved in proposition \ref{prop:square-loss-r-vs-b}, we could test the equivalent condition, $r = b$.
    If we reject the hypothesis, then we would have evidence that the admissions decision relies on race. The benefit of this metric over something like a coefficient on the race variable is that this approach is agnostic to the specific modelling assumptions on the conditional expectation function; therefore, it can better accommodate flexible ML methods that excel in high-dimensional settings. 
\end{example}

\begin{example}[Preventing Manipulation]
    Still in the admissions setting, we might want to ensure a ranking among some subset of the covariates. For example, in order to prevent manipulation of the admissions process, we might require that the admissions decision is less sensitive to self-reported community service hours than to exam scores. In this case, we would test $r(\text{community service hours}) < r(\text{exam score})$. An advantage of this framework is that it enables us to compare any two variables, regardless of their units.
\end{example}

\begin{example}[Idealized Baselines]
    We might instead want to see how the admissions officer's behavior compares to an idealized decision rule that we cook up ourselves. For example, we might calibrate a simple admissions rule to our own preferences and generate a ranking of variables by their importance. We would then do the same for the observed admissions data and check for deviations between the two rankings. We could also apply proposition \ref{prop:equivalent-ranking} to directly verify that the admissions officer relies less on a sensitive variable like race than the idealized rule does.
\end{example}

\subsection{Alternative Formulations}

I briefly mention two alternatives to the formulation of reliance in definition \ref{def:reliance}. These alternatives share a similar motivation to the original definition, but their definition of `noise' differs.

First, in our motivation for reliance, we made $X_1$ noise by making it completely uninformative of $(Y, X_2)$. That is, we shuffled in $X_1^b \perp (Y^a, X_2^a)$. We might instead assert that $X_1$ adds no additional information on $Y$ given $X_2$. In other words, we might only need $X_1^b \perp Y^a \mid X_2^a$. This \textbf{conditional reliance} measure instead takes one draw $(Y^a, X_1^a, X_2^a) \sim P$ and picks $X_2^b \sim P_{X_2^b \mid X_1^b = X_2^a}$. 
\citet{Fisher} present this definition too in section 8.2.

Second, we might measure the \textbf{worst-case reliance}
$$r = \sup_{x_1^b \in \mathcal{X}_1} \E L( Y^a, f(x_1^b, X_2^a) )$$
where the expectation is over $(Y^a, X_2^a) \sim P_{Y,X_2}$. The max replaces $X_1$ with the noise that creates the largest loss.
Compared to these two alternatives, the one in definition \ref{def:reliance} is much easier to compute. Furthermore, the conditional reliance measure will suffer in small samples and high-dimensional settings since it is less likely to find multiple observations with the same $X_2$ to shuffle together. 

\subsection{Estimating Reliance}

If we can estimate counterfactual queries to the oracle, then we can define a plug-in estimator for $r$:

\begin{definition}[Estimator for $r$]
    \label{def:estimator-for-r}
    Given an estimator $\hat{f} : \mathcal{X}_1 \times \mathcal{X}_2 \rightarrow \mathcal{Y}$ for $f$, define the estimator
    $$\hat{r} = \frac{1}{n(n-1)}\sum_{i = 1}^n \sum_{j \ne i} L( y_i, \hat{f}(x_{1j}, x_{2i}) )$$
    over the data $\{ (y_i, x_{1i}, x_{2i}) \}_{i=1}^n$.
\end{definition}

\begin{remark}
    Assuming that the loss has finite variance, $\hat{r}$ has a normal limiting distribution with mean $\E[L(Y^a, \hat{f}(X_1^b, X_2^a))]$ by the central limit theorem. $\hat{r}$'s distribution depends on that of $\hat{f}$.
\end{remark}

The double sum can make this object computationally expensive: if evaluating $\hat{f}$ is $O(1)$, then directly computing this double sum is $O(n^2)$. However, if $X_1$ realizes few values, we can compute $\hat{r}$ much more efficiently:

\begin{proposition}
    \label{prop:efficient_r_hat_by_categories}
    Let $C = \{ x_{1i} \}_{i=1}^n$ be the distinct observed values that $X_1$ takes in the data. Then,
    $$\hat{r} = \frac{1}{n(n-1)}\sum_{i=1}^n \sum_{c \in C} (n_c - 1\{x_{1i} = c\})L(y_i, \hat{f}(c, x_{2i})),$$
    where $n_c = |\{ i : x_{1i} = c \}|$.
\end{proposition}

\begin{proof}\hyperlink{pf:efficient_r_hat_by_categories}{Proof in appendix.}\end{proof}

\begin{remark}
    We can implement this second formulation of $\hat{r}$ in $O(n|C|)$ time complexity, assuming that evaluating $\hat{f}$ is $O(1)$.
    To do this, precompute $n_c$ for each $c \in C$ by looping through the $x_{1i}$'s and tracking counts in a hash table whose keys are $C$; this step is $O(n)$.
\end{remark}

\begin{remark}
    As a corollary, if $X_1$ is binary, i.e., $C = \{0, 1\}$, then computing $\hat{r}$ using the formula in proposition \ref{prop:efficient_r_hat_by_categories} is only $O(n)$. In fact, it is $O(n)$ as long as $X_1$ has finite support (though the hidden constant can be quite large).
\end{remark}

The oracle asks what the decision-maker would have done were he presented with a specific vector of covariates.
We rely on the oracle because we cannot directly query the decision-maker, but this counterfactual inference incurs the cost of additional assumptions and statistical uncertainty.
In fact, this object might only be partially identified:
black-box decision-makers often rely on private information, so $X_2$ may be only partially observed, or some data might be systematically missing.
If counterfactual queries are only partially identified, then we may need to settle for bounds on $r$:

\begin{proposition}
    \label{prop:conservative-bounds}
    Suppose $\mathcal{Y} = [0,1]$, but $y_i \in \{0,1\}$ for all $i$, and assume that $L(y, \hat{y})$ increases monotonically with $|y - \hat{y}|$.
    As a normalization, assume $L(y, \hat{y}) = 0$ if $|y - \hat{y}| = 0$.
    Then, given bounds $\hat{f}_{min}, \hat{f}_{max}$ on $\hat{f}$, we can obtain the conservative bounds
    $$\hat{r}_{min} = \frac{1}{n(n-1)} \sum_{i=1}^n \sum_{j \ne i} \min\{ L( y_i, \hat{f}_{min}(x_{1j}, x_{2i}) ), L( y_i, \hat{f}_{max}(x_{1j}, x_{2i}) ) \}$$
    and $\hat{r}_{max}$ by replacing $\min$ with $\max$. 
    That is, $\hat{r} \in [\hat{r}_{min}, \hat{r}_{max}]$.
\end{proposition}

\begin{proof}\hyperlink{pf:conservative-bounds}{Proof in appendix.}\end{proof}

\section{Example: School Admissions}
\label{sec:school-admissions}

\subsection{Setup}

This worked example shows how this framework could be applied end-to-end:
an admissions officer decides if to admit a student depending on his or her race, sex, and test score. 
The choice $Y \mid X$ is characterized by the decision rule
$$Y = 1\{ -2X_1 + X_2 + X_3/5 - 2.2 \geq 0 \}$$
where $X_1 \sim Bernoulli(.5)$ and $X_2 \sim Bernoulli(.3)$. Also, we define $X_3 = A + E$, rounded to the nearest integer between 0 and 10, where innate ability $A \sim N(6,1)$ and study effort $E \sim N(1,1)$.

Students self-report if they were admitted. However, they only respond to the survey if $Z=1 \mid X,E$, where
$$Z = 1\{ X_1 + 3X_2 + X_3/8 + E - 3.5 \geq 0 \}.$$
The researchers want to estimate the admissions officer's reliance on $X_1,X_2$ and $X_3$. Importantly, the researchers know the survey response rate ($E[Z]$) and assume $Y \mid Z, X_1, X_2, X_3$ is roughly logistically distributed; however, they do not make distributional assumptions about the noise, nor do they know that, in fact, $Y \perp Z \mid X_1, X_2, X_3$.  

\subsection{Identification}

To compute these reliance values, we must estimate $P(Y=1 \mid X)$, but, given the researchers' knowledge, there is confounding in the missing data problem due to $E$. Therefore, this estimand is only partially identified.
We only observe $Y,X$ when $Z=1$ and we have $E[Z]$, so our identification result is
\begin{equation}
    \label{eq:school-admissions-identification}
    E[Y \mid X] \in [ E[Y \mid X,Z=1]P(Z=1), \; E[Y \mid X,Z=1]P(Z=1) + P(Z=0) ].
\end{equation}
The interval is smaller when $P(Z=1)$ is larger, and the interval is just a point when $P(Z=1) = 1$.

\begin{proposition}
    \label{prop:school-admissions-tight-identification}
    The identification region in equation \ref{eq:school-admissions-identification} is tight.
\end{proposition}

\begin{proof}\hyperlink{pf:school-admissions-tight-identification}{Proof in appendix.}\end{proof}

\subsection{Estimation}

In this simulation, $10,000$ students apply. 8,313 students respond to the survey, and their acceptance rate is 13\%.

Estimate $f_1(x) = E[Y \mid X=x,Z=1]$ by logistic regression of the observed $Y$ on $X$; denote this estimator by $\hat{f}_1$.
Even though $Y \mid X,Z=1$ isn't necessarily logistic, over 99.8\% of the sample is classified correctly, and this approximation will suffice for the sake of this example.
We can now use our identification result in equation \ref{eq:school-admissions-identification} to define lower and upper bounds for $f(x) = E[Y \mid X=x]$:
\begin{align*}
    f_{min}(x) &= f_1(x)P(Z=1) \quad\text{and} \\
    f_{max}(x) &= f_1(x)P(Z=1) + P(Z=0).
\end{align*}
Then, estimate $\hat{f}_{min}$ and $\hat{f}_{max}$ by plugging in $\hat{f}_1$ for $f_1$.

\subsection{Compute Reliance}

For measuring the reliance on $X_k$, we plug these estimators into the result from proposition \ref{prop:conservative-bounds} to obtain the functions:
\begin{align*}
    \hat{r}_{min}^1(k) &= \frac{1}{n(n-1)} \sum_{i=1}^n \sum_{j \ne i} \min\{ L( y_i, \hat{f}_{min}(s(i,j,k)) ), \,\, L( y_i, \hat{f}_{max}(s(i,j,k)) ) \} \\
    \hat{r}_{max}^1(k) &= \frac{1}{n(n-1)} \sum_{i=1}^n \sum_{j \ne i} \max\{ L( y_i, \hat{f}_{min}(s(i,j,k)) ), \,\, L( y_i, \hat{f}_{max}(s(i,j,k)) ) \}
\end{align*}
where the shuffling function $s(i,j,k)$ returns row $i$ with $x_{jk}$ shuffled into the $k$th slot. 
However, these bounds are not bounds on $r$. For example, for $k = 1$, the observed data are $X_1 \mid Z=1$ and $(Y,X_2,X_3) \mid Z=1$, hence the superscript 1 on the variable names. 
That is, $[\hat{r}_{min}^1(1), \hat{r}_{max}^1(1)]$ bounds how much the admissions officer depends on race for students who respond to the survey.

\begin{figure}[htbp]
    \centering
    \begin{minipage}{\textwidth}
        \begin{adjustbox}{width=\textwidth,center}
            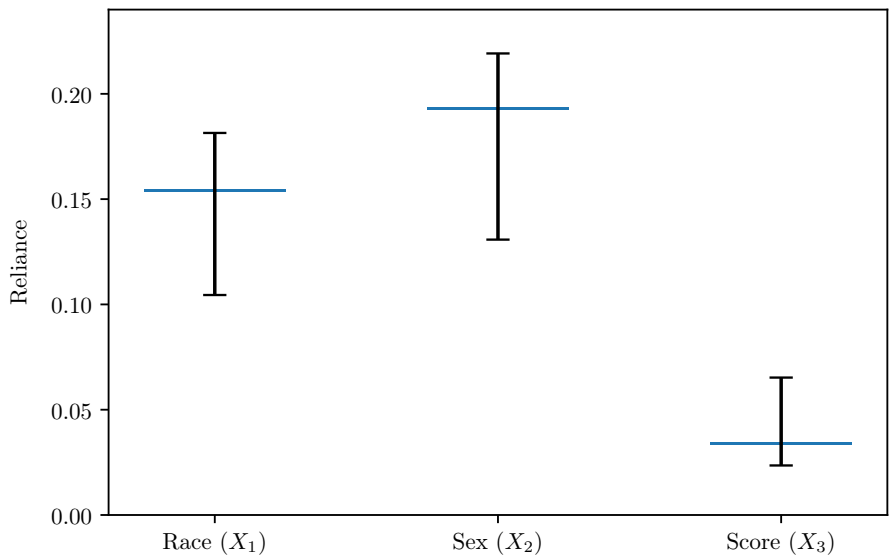
        \end{adjustbox}
        {\footnotesize Note: Blue lines indicate true reliance values.}
    \end{minipage}
    \caption{Reliance for Responding Students}
    \label{fig:mr_z1}
\end{figure}

Using the square loss $L(y, \hat{y}) = (y - \hat{y})^2$, we obtain the bands in figure \ref{fig:mr_z1}.
Observe that the bands for race and sex are completely above the band for test score, but the bands for race and sex overlap.
Thus, we might conclude that, for students who ultimately respond to the survey, the admissions officer relies more on race and sex than on test scores. But, it is inconclusive if she relies more on race or on sex.
That is, for an average student who responds to the survey, replacing race or sex with noise changes the admissions decision more than replacing the test score with noise. If we believe, however, that test score should be the most important variable for this (self-selecting) cohort, then we might call for a closer look at or revision of admissions practices.

Note that the ranking of reliance values differs from that of the coefficients in the admissions decision rule, $Y = 1\{ -2X_1 + X_2 + X_3/5 - 2.2 \geq 0 \}$. Here, $X_1$ has coefficient $-2$, which is larger in absolute value than the coefficient of $1$ on $X_2$, but the reliance on $X_2$ is higher. Furthermore, we cannot compare the coefficients on race or sex to the coefficient on test score; we can, however, directly compare the reliance on them.

\section{Example: Interruptions During Supreme Court Oral Argument}
\label{sec:supreme-court-interruptions}

\subsection{Introduction}

Using Supreme Court oral argument transcripts since 1982 \citep{Chang-ConvoKit, Danescu-Niculescu-Mizil+al:12a}, \citet{cai_et_al_2023} measure gender's effect on how often Justices interrupt advocates. 
An oral argument is comprised of a sequence of utterances, each with one speaker. The authors extract \textit{chunks} from each oral argument. A chunk is a contiguous dialogue of four or more utterances between exactly one advocate and one Justice. For example, one chunk from \textit{Comcast Corp. v. National Association of African American-Owned Media} is:

\begin{quote}
    \tt
    Erwin Chemerinsky:
    If at the end the plaintiff concedes that he or she would have never gotten the contract anyway, I believe, at the end, under the standard adopted in Patterson versus McLean, the plaintiff would not prevail.
    
    Justice John G. Roberts Jr.:
    So that the --
    
    Erwin Chemerinsky:
    But that doesn't --
    
    Justice John G. Roberts Jr.:
    I'm sorry.
    Go ahead.
    
    Erwin Chemerinsky:
    I was going to say but that doesn't tell us what's required at the pleading stage or at the prima facie case stage.
\end{quote}

This chunk has 5 utterances, Erwin Chemerinsky is the advocate, and Justice Roberts is the speaker. 
As Cai et al. note, these transcripts are manually typed and consistently formatted, and interruptions are indicated with either two dashes (as in this chunk) or two dots at the end of an utterance. 
In this chunk, the advocate says 62 words (\textit{advocate tokens}) and is interrupted by Justice Roberts twice. In the entire dataset of 677,294 chunks, there is a mean of 59 tokens per chunk, and the median is 28. 

Cai et al. seek to estimate the effect of gender $G_i$ on the \textit{token-normalized interruption rate} $Y_{i \mid j}$ of chunk $i$ with Justice $j$:
\begin{equation}
    \label{eq:interruption-rate}
    Y_{i \mid j} = \frac{\text{number of interruptions by Justice $j$ in chunk $i$}}{(\text{number of advocate tokens in chunk $i$}) / 1000},
\end{equation}
or the number of times the Justice would interrupt the advocate if the advocate spoke 1,000 words, which is about 4 pages of 12-point font, double-spaced text.
The Justices' median interruption rates range from Justice Blackmun's 2.2 interruptions per thousand tokens to Justice Breyer's 11.0. The median interruption rate overall is 6.8, and the mean is 10.7. 

The authors assume that the interruption rates of all chunks, i.e., the observations, are independent and that ``there is no unmeasured confounding.'' 
They consider but ultimately decline to control for the \textit{ideological alignment} between the Justice and the advocate's argument, the advocate's \textit{stylistic quality}, and the advocate's \textit{experience}.
Reformulated in Rubin's potential outcomes framework, $(Y_{i \mid j}^{(0)}, Y_{i \mid j}^{(1)}) \perp G_i$, where $Y_{i \mid j}^{(g)}$ is the interruption rate of chunk $i$ by Justice $j$ if the advocate's gender in that chunk were $g$.
Given these assumptions, the authors compute the difference between the mean male and female interruption rates for each Justice and conclude that there is a ``clear and consistent gender effect.''

My analysis contributes new results: First, recognizing the presence of outliers and heteroskedasticity and controlling for argumentative alignment and advocate experience, I find that the effect of gender is more muted by estimating a robust liner regression. Second, I compute cross-distribution reliance values for these three covariates, as in definition \ref{def:cross-distribution-reliance}, allowing us to compare the importance of otherwise incommensurable covariates and illustrating the value of this paper's framework.

\subsection{Estimate the Effect of Gender}

Like Cai et al., I divide each oral argument since 1982 into chunks of four or more utterances with one Justice and one advocate, and I compute each chunk's token-normalized interruption rate as in equation \ref{eq:interruption-rate}. 

Like them, I determine each advocate's gender by checking the honorific with which the Justices address that advocate in the oral argument. For example, for Erwin Chemerinsky, I check if the Justices ever say ``Mr. Chemerinsky'' or ``Ms. Chemerinsky.'' If neither honorific matches, for example because the speaker is addressed as ``General,'' I use the confident classifications of an open-source gender guesser \citep{gender-guesser}. I manually resolve a few names that match both honorifics; for example in \textit{Pierce v. Underwood}, Justice Rhenquist accidentally calls advocate Mary Burdick ``Mr. Burdick.'' I drop any advocates without a matched gender. The full pipeline is available online \citep{vebman-scorpus}.

I define \textit{experience} as the number of oral arguments since 1982 in which a particular advocate appears. 
I also define \textit{alignment} as whether or not the Justice ultimately votes for that advocate's side. Judges, clerks, and scholars doubt that oral argument actually changes decisions in all but the closest cases \citep{Wolfson, Duvall, Coleman}, mitigating the possibility that argumentative quality affects both interruptions and the Justice's ultimate decision. Even though alignment is weakly correlated with gender (--0.00375), it may still be relevant because Justices use oral arguments to refine their opinions.
Ultimately, the identification assumption is
$$(Y_{i \mid j}^{(0)}, Y_{i \mid j}^{(1)}) \perp \textit{Gender}_i \mid \textit{Experience}_i, \textit{Alignment}_i.$$

\begin{figure}[htbp]
    \centering
    \begin{adjustbox}{width=1.5\textwidth,center}
        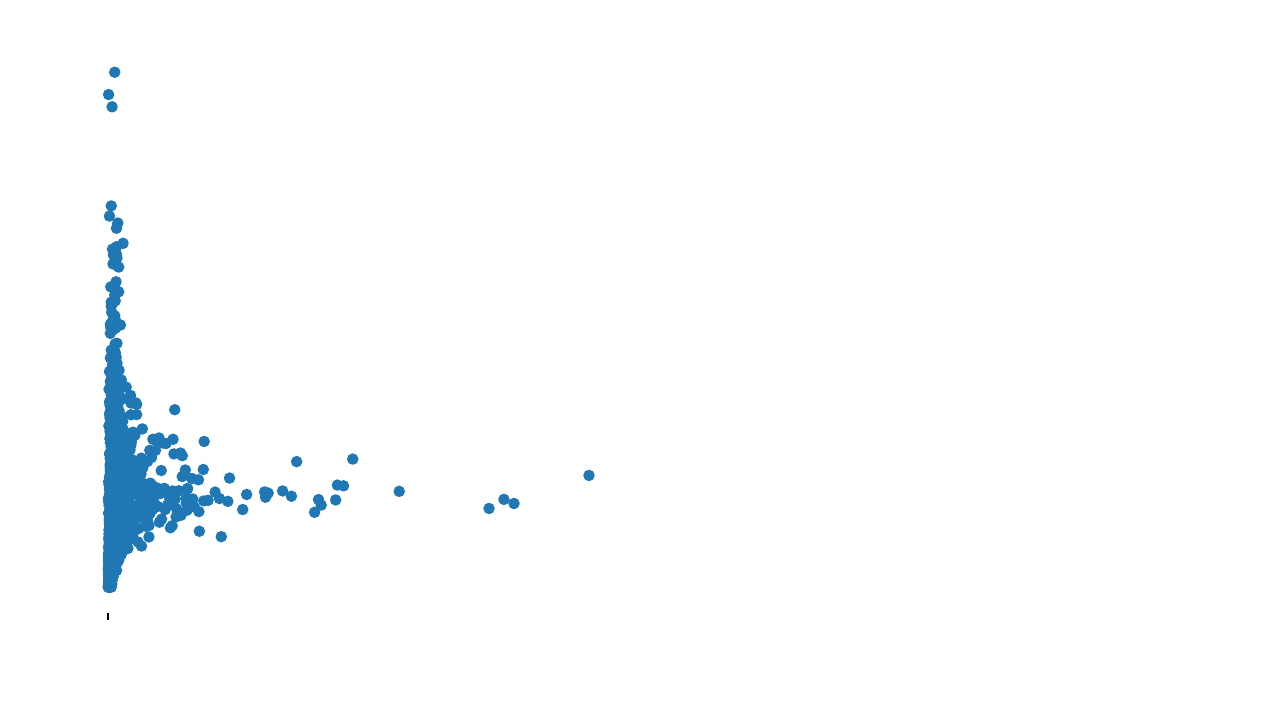
    \end{adjustbox}
    \begin{minipage}{\textwidth}
        {\footnotesize Note: Each point represents a distinct advocate.
        An advocate's mean interruption rate is the average of the interruption rates over all chunks with that advocate.}
    \end{minipage}
    \caption{Advocate Mean Interruption Rate vs. Number of Chunks}
    \label{fig:scorpus-interruptions-per-chunks}
\end{figure}

I use a robust regression because the data are heteroskedastic and contain outliers. A White test for heteroskedasticity rejects the null hypothesis at the 5\% level for 6 out of 21 Justices. 
Furthermore, figure \ref{fig:scorpus-interruptions-per-chunks}, which shows each advocate's mean interruption rate over the number of chunks with that advocate, helps reveal that there are outliers. There is significantly more variance among advocates who appear in fewer chunks, and there are clear outliers among them. For example, Lisa Corkran is the advocate with the highest mean interruption rate of 81 interruptions per thousand tokens; however, she appears in only 16 chunks, including one with Justice Breyer and an interruption rate of 200 and two with Justice Roberts and an interruption rate of 500. 

\begin{figure}[htbp]
    \centering
    \begin{adjustbox}{width=1.5\textwidth,center}
        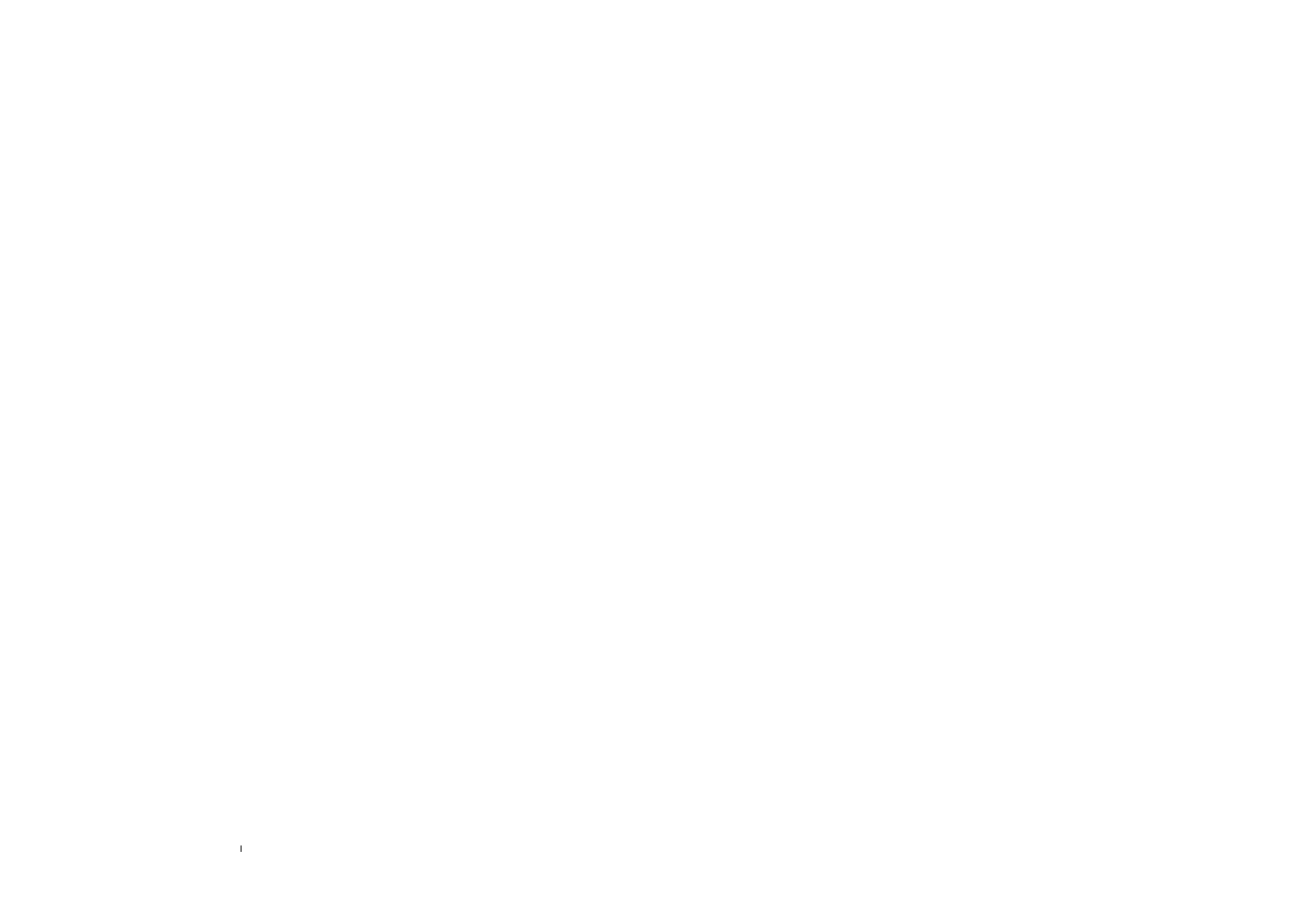
    \end{adjustbox}
    \begin{minipage}{\textwidth}
        {\footnotesize Note: Huber-estimator robust regression coefficient in blue and non-robust (vanilla OLS) in gray. Both regressions control for alignment and experience.
        95\% confidence intervals shown.
        Observed difference in means in red.}
    \end{minipage}
    \caption{Average Treatment Effect of Being Female}
    \label{fig:scorpus-ate}
\end{figure}

Figure \ref{fig:scorpus-ate} displays the results. For each Justice, we estimate a robust regression of each chunk's interruption rate on the advocate's gender, experience, and argumentative alignment.
The average treatment effect of gender computed by the robust regression is in blue, and that computed with vanilla OLS is in gray. The observed difference between mean male and female interruption rates is in red. 
The observed difference is often closer to 0 compared to the vanilla OLS ATE, and the robust estimates are consistently closer to 0 than both. Notably, the sign of Justice Marshall's estimate flips when using the robust estimate. 

Overall, 10 out of 21 Justices have robust confidence intervals that are entirely greater than 0. That is, for 10 Justices, being a woman increases the number of interruptions at the 5\% significance level. For Justice Powell only, being a woman decreases the number of interruptions at the 5\% significance level. Although statistically significant, the magnitudes are modest: the robust ATE of gender is less than 1 interruption per thousand tokens in magnitude for 12 Justices. For 19 Justices, being female adds (or, in Kavanaugh and Powell's case, subtracts) fewer than 2 interruptions per thousand tokens.

\subsection{Compute Reliance Values}

We now directly apply this paper's definitions and results to compute each Justice's reliance on gender, experience, and alignment.
We use the square loss $L(y, \hat{y}) = (y - \hat{y})^2$ and let $\hat{f}_j$ be the prediction function of the Huber estimator \citep{huber-robust-rlm} for Justice $j$ used to compute the robust ATE above.
We compute the within-distribution reliance for each Justice as in proposition \ref{prop:efficient_r_hat_by_categories}, which provides massive performance improvements over the formulation in definition \ref{def:estimator-for-r}.
To allow comparisons across Justices, we must normalize the within-distribution reliance values. If we define an additively separable total loss function over the $J$ Justices
$$\mathcal{L}( (y_j)_{j=1}^J, (\hat{y}_j)_{j=1}^J ) = \sum_{j=1}^J L(y_j, \hat{y}_j)$$
then we can apply proposition \ref{prop:equivalent-ranking} to obtain cross-distribution rankings of reliance values by subtracting off an estimate of the baseline reliance, 
$$\hat{b}_j = \frac{\sum_{\text{chunk } i \text{ with Justice } j} L(y_{i \mid j}, \hat{f}(\textit{Gender}_i, \textit{Experience}_i, \textit{Aignment}_i))}{\text{number of chunks with Justice $j$}}.$$

\begin{figure}[htbp]
    \centering
    \begin{adjustbox}{width=1.5\textwidth,center}
        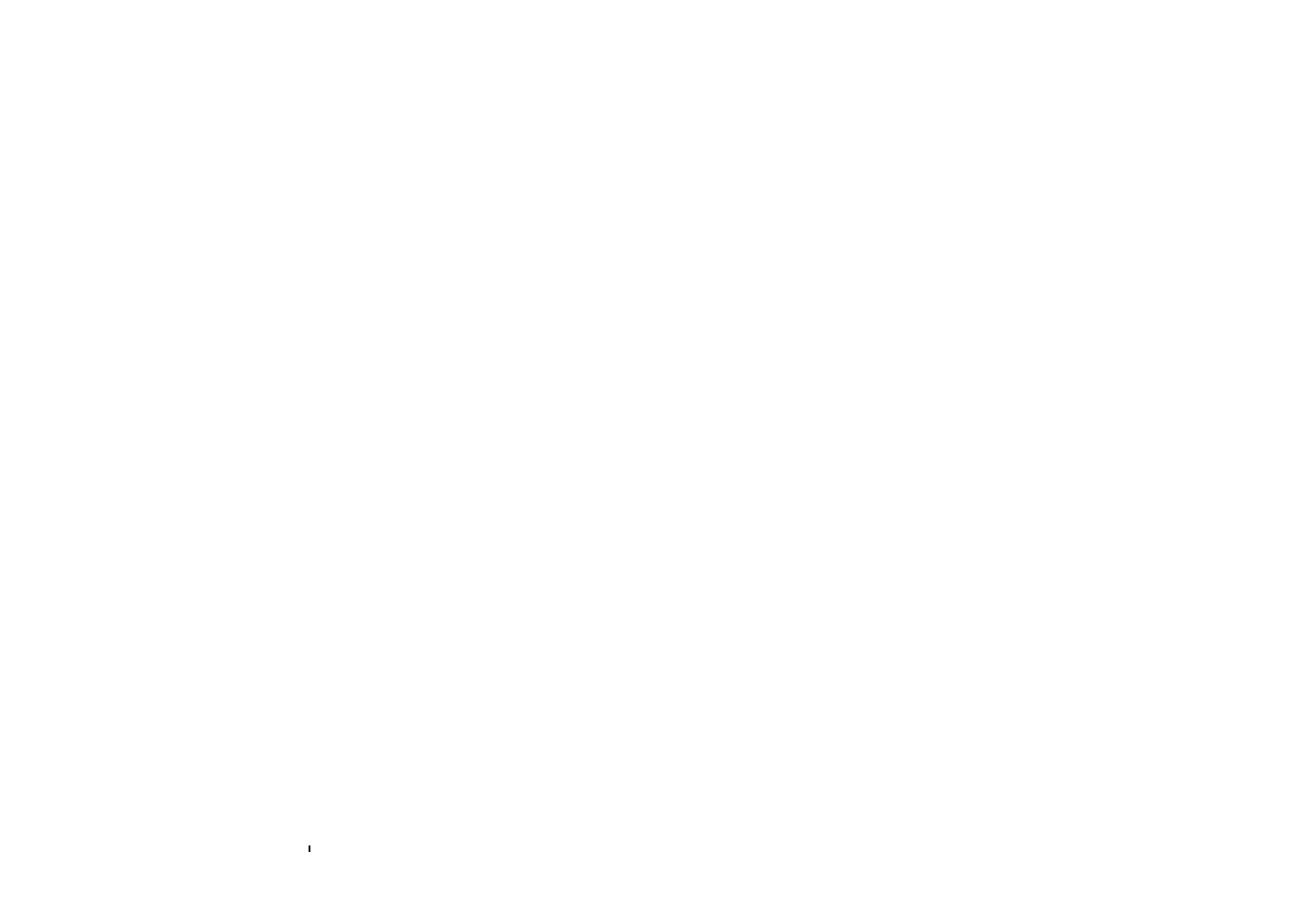
    \end{adjustbox}
    \begin{minipage}{\textwidth}
        {\footnotesize Note: Reliance on gender (red), alignment (green), and experience (blue). Justices in the same horizontal panel rely on these variables in the same order. Colored lines connect points for ease of viewing.}
    \end{minipage}
    \caption{Reliance on Gender, Alignment, and Experience}
    \label{fig:scorpus-reliance}
\end{figure}

The results are summarized in figure \ref{fig:scorpus-reliance}; the raw reliance values and Huber estimate coefficients are also in table \ref{table:reliance-and-coefficients} in the appendix.
In example \ref{ex:square-loss}, we showed that $r \geq b \geq 0$, and hence our reliance values, $r - b$, must be non-negative. The presence of negative estimated reliance indicates that $\hat{f}$ produces errors, which propagate through $\hat{r}$ and $\hat{b}$. 

The Huber-ATE and reliance (both recorded in table \ref{table:reliance-and-coefficients}) provide competing measures of how much gender and alignment matter to each Justice:
for 10 out of 21 Justices, the ATE of gender is greater than the ATE of alignment.
These same 10 Justices, in addition to Kagan, Marshall, and Ginsburg, rely more on gender than on alignment (these are the top 3 panels in figure \ref{fig:scorpus-reliance}). That being said, the point estimates of reliance on gender and alignment are very close for many Justices.

We cannot directly compare the coefficients on gender and alignment to the coefficient on experience. Using reliance, however, we can directly compare the importance of any variables:
\begin{itemize}
    \item 11 Justices rely the most on gender, 7 rely the most on alignment, and 3 rely the most on experience.
    \item 13 Justices rely more on gender than on experience when `deciding' if to interrupt an advocate, but 8 Justices rely more on experience. 
    \item Justices Brennan, Roberts, and Scalia rely the most on gender among all the Justices. Justices Burger, Sotomayor, and Marshall rely the least. (Thomas's estimate is unreliable because he very rarely interrupts.)
    \item Of the Justices currently on the Court, only Roberts (and Thomas) rely the most on gender. Alito, Gorsuch, and Sotomayor rely the least on gender among all the Justices; instead, they all rely the most on alignment. 
    Kagan and Kavanaugh rely the most on experience.
    \item Sotomayor relies more on alignment than Roberts relies on gender, even though he relies the most on gender among any Justice on the Court (besides Thomas).
\end{itemize}

\section{Conclusion}

This paper expands and expounds on the permutation-based measure of variable importance, inspired by \citet{Fisher} and earlier used by \citet{Breiman, Breiman_Manual} for random forests.
For economics and the social sciences, the approach discussed in this paper introduces a flexible and interpretable framework to quantify how much black-box decision-makers rely on variables of interest. I discuss some of the considerations in implementing such a framework, rigorously connect it to the machine learning fairness literature, explain how reliance computations can have policy implications, and present illustrative and applied examples.

This work also contributes to the machine learning explainability literature by incorporating counterfactual estimation, including partial identification of counterfactual queries, from economics. Partial identification is particularly new. Settling for bounds instead of point-estimates can allow for more credible assumptions about how black boxes operate under the hood. As machine learning models become more broadly deployed, including in proprietary contexts, auditors' access to these models will decrease. With credible assumptions, it might be possible to estimate counterfactual queries when directly querying the black box is impossible. 

This paper equips the analyst with a new, flexible, and intuitive method for understanding how decision-makers think. By helping us explain opaque decisions, it has the potential to ensure fairness, confirm priorities, and improve outcomes across a huge array of contexts.

\section*{Appendix}
\label{sec:appendix}

\begin{proof}[Proof of proposition \ref{prop:square-loss-r-vs-b}]
    \hypertarget{pf:square-loss-r-vs-b}{}
    \textbf{First, prove $r \geq b$ and 1 $\Leftrightarrow$ 2.}
    We prove the claim by reformulating $b$ as the solution to the optimization of mean squared prediction error (MSPE) over the vector $(Y^a, X_1^a, X_1^b, X_2^a)$, which is the coupling where $(Y^a, X_1^a, X_2^a) \sim P$ is independent of $X_1^b \sim P_{X_1}$:
    \begin{align*}
        r 
        &= \E_{Y^a, X_1^b, X_2^a} (Y^a - f(X_1^b, X_2^a))^2 \\
        &= \E_{Y^a, X_1^a, X_1^b, X_2^a} (Y^a - f(X_1^b, X_2^a))^2 \\
        &\geq \min_{\substack{g : \mathcal{X}_1^2 \times \mathcal{X}_2 \rightarrow \mathbb{R} \\ \text{measurable}}} \E_{Y^a, X_1^a, X_1^b, X_2^a} (Y^a - g(X_1^a, X_1^b, X_2^a))^2.
    \end{align*} 
    The conditional expectation minimizes MSPE, so this minimum is attained by
    \begin{align*}
        g(x_1^a, x_1^b, x_2^a) 
        &= \E[ Y^a \mid X_1^a = x_1^a, X_1^b = x_1^b, X_2^a = x_2^a ] \\
        &= \E[ Y^a \mid X_1^a = x_1^a, X_2^a = x_2^a ] \\
        &= \E[ Y \mid X_1 = x_1^a, X_2 = x_2^a ] \\
        &= f(x_1^a, x_2^a).
    \end{align*}
    The second equality holds because $(Y^a, X_1^a, X_2^a) \perp X_1^b$. 
    The third equality holds because $(Y, X_1, X_2)$ and $(Y^a, X_1^a, X_2^a)$ are both distributed according to $P$.
    Thus:
    $$r \geq \E_{Y^a, X_1^a, X_1^b, X_2^a} (Y^a - f(X_1^a, X_2^a))^2 = b.$$
    The second claim follows because the minimizer of MSPE is unique almost surely.

    \textbf{Second, prove 2 $\Leftrightarrow$ 3.}
    We will prove the forward implication \textit{(only if)} by the contrapositive:
    Suppose $f(X_1^a, X_2^a) = f(X_1^b, X_2^a)$ does not hold almost surely.
    Therefore, there exists $U \subseteq \mathcal{X}_1^2 \times \mathcal{X}_2$ such that $P_{X_1^a,X_1^b,X_2^a}(U) > 0$ and $f(x_1^a, x_2^a) \ne f(x_1^b, x_2^a)$ for all $(x_1^a, x_1^b, x_2^a) \in U$.
    Decompose $U$ into two potentially overlapping subsets 
    $$U^a = \{ (x_1^a,x_1^b,x_2^a) \in U : f(x_1^a,x_2^a) \ne f(x_2^a) \}$$
    and
    $$U^b = \{ (x_1^a,x_1^b,x_2^a) \in U : f(x_1^b,x_2^a) \ne f(x_2^a) \}.$$
    Observe that for $u = (x_1^a,x_1^b,x_2^a) \in U$, if $u \notin U^a \cup U^b$ then 
    $$f(x_1^a,x_2^a) = f(x_2^a) = f(x_1^b,x_2^a),$$ 
    contradicting that $u \in U$. Thus, $U = U^a \cup U^b$. 
    Therefore, $P_{X_1^a,X_1^b,X_2^a}(U) > 0$ implies $P_{X_1^a,X_1^b,X_2^a}(U^a) > 0$ or $P_{X_1^a,X_1^b,X_2^a}(U^b) > 0$ (or both). 
    \begin{enumerate}[a.]
        \item 
            If $P_{X_1^a,X_1^b,X_2^a}(U^a) > 0$, 
            then define $V = \{ (x_1^a,x_2^a) : (x_1^a,x_1^b,x_2^a) \in U \}$
            by dropping $x_1^b$ from the vector.
            We thus have $P_{X_1^a,X_2^a}(V) > 0$, so $f(X_1^a, X_2^a) = f(X_2^a)$ does not hold almost surely.

        \item 
            If $P_{X_1^a,X_1^b,X_2^a}(U^b) > 0$, 
            then define $V = \{ (x_1^b,x_2^a) : (x_1^a,x_1^b,x_2^a) \in U \}$
            by dropping $x_1^a$ from the vector.
            We thus have $P_{X_1^b,X_2^a}(V) > 0$, which, by absolute continuity, implies $P_{X_1^a,X_2^a}(V) > 0$.
            Hence, $f(X_1^a, X_2^a) = f(X_2^a)$ does not hold almost surely.
    \end{enumerate}
    Thus, we have proved the forward direction by the contrapositive.
    
    Now, prove the reverse implication \textit{(if)} also by the contrapositive:
    Suppose that $f(X_1^a, X_2^a) = f(X_2^a)$ does not hold almost surely.
    Therefore, there exists a set $A \subseteq \mathcal{X}_1 \times \mathcal{X}_2$ such that $P_{X_1^a,X_2^a}(A) > 0$ and for all $(x_1^a, x_2^a) \in A$,
    \begin{align*}
        f(x_1^a, x_2^a) &\ne f(x_2^a) \\
        &= \E[Y \mid X_2 = x_2^a] \\
        &= \E_{X_1} [\E[Y \mid X_1, X_2 = x_2^a] \mid X_2 = x_2^a] \\
        &= \E_{X_1}[f(X_1, x_2^a) \mid X_2 = x_2^a].
    \end{align*}
    The second equality holds by the law of iterated expectation.
    Furthermore,
    $$B(x_1^a, x_2^a) = \{ x_1^b \in \mathcal{X}_1 : f(x_1^a, x_2^a) \ne f(x_1^b, x_2^a) \}$$
    satisfies 
    $$P_{X_1 \mid X_2 = x_2^a}(B(x_1^a, x_2^a)) > 0,$$
    otherwise the inequality would fail for a given $(x_1^a, x_2^a)$ pair. 
    Therefore, each $B(x_1^a, x_2^a)$ also has positive probability with respect to the unconditioned distribution $P_{X_1}$; that is, $P_{X_1}(B(x_1^a, x_2^a)) > 0$. Since $P_{X_1^b} = P_{X_1}$, we have 
    $$P_{X_1^b}(B(x_1^a, x_2^a)) > 0$$
    for all $(x_1^a, x_2^a) \in S$.
    Now, combine $A$ and each $B(\cdot, \cdot)$ to produce the set
    $$U = \{ (x_1^a, x_1^b, x_2^a) : (x_1^a, x_2^a) \in A \wedge x_1^b \in B(x_1^a, x_2^a) \}.$$
    Then,
    $$P_{X_1^a,X_1^b,X_2^a}(U) = P_{X_1^a,X_2^a}(A) P_{X_1^b}\left(\bigcup_{(x_1^a, x_2^a) \in A} B(x_1^a, x_2^a)\right)$$
    because $(X_1^a,X_2^a) \perp X_1^b$. The first term has positive probability by assumption. The second term also has positive probability because the union is nonempty by assumption and each $B(x_1^a, x_2^a)$ has positive probability.

    There thus exists a set $U \subseteq \mathcal{X}_1^2 \times \mathcal{X}_2$ such that $P_{X_1^a,X_1^b,X_2^a}(U) > 0$ and $f(x_1^a, x_2^a) \ne f(x_1^b, x_2^a)$ for all $(x_1^a, x_1^b, x_2^a) \in U$.
    Hence, $f(X_1^a, X_2^a) = f(X_1^b, X_2^a)$ does not hold almost surely, and we have proved the second implication by the contrapositive.
\end{proof}

\begin{proof}[Proof of Proposition \ref{prop:equivalent-ranking}]
    \hypertarget{pf:equivalent-ranking}{}
    Because $L$ is additively separable and the expected value commutes with addition, we can decompose cross-distribution reliance $r_k^\times$ into
    \begin{align*}
        r_k^\times 
        &= \sum_{i=1}^n \E L_i(Y_i^a, f_i(s_k(X_{1i}^a, X_{1i}^b), X_{2i}^a)) \\
        &= \E L_k(Y_k^a, f_k(X_{1k}^b), X_{2k}^a) + \sum_{\substack{i = 1 \\ i \ne k}}^n \E L_i(Y_i^a, f_i(X_{1i}^a, X_{2i}^a)) \\
        &= r_k + \sum_{\substack{i = 1 \\ i \ne k}}^n b_i.
    \end{align*}
    Subtracting $\sum_{i=1}^n b_i$ from both sides gives
    $$r_k^\times - \sum_{i=1}^n b_i = r_k - b_k$$
    Thus, $r_k - b_k$ gives the same ranking as $r_k^\times - \sum_{i=1}^n b_i$. This ranking shifts all $r_k^\times$'s by the same constant $\sum_{i=1}^n b_i$, and hence gives the same ranking as $r_k^\times$.
\end{proof}

\begin{proof}[Proof of Proposition \ref{prop:efficient_r_hat_by_categories}]
    \hypertarget{pf:efficient_r_hat_by_categories}{}
    Taking the double sum from definition \ref{def:estimator-for-r}:
    \begin{align*}
        \sum_{i = 1}^n \sum_{j \ne i} L( y_i, \hat{f}(x_{1j}, x_{2i}) )
        &= \sum_{i = 1}^n \sum_{c \in C} \sum_{\substack{j \ne i \\ x_{1j} = c }} L( y_i, \hat{f}(x_{1j}, x_{2i}) ) \\
        &= \sum_{i = 1}^n \sum_{c \in C} \sum_{\substack{j \ne i \\ x_{1j} = c }} L( y_i, \hat{f}(c, x_{2i}) ) \\
        &= \sum_{i = 1}^n \sum_{c \in C} \left(  L( y_i, \hat{f}(c, x_{2i}) ) \sum_{\substack{j \ne i \\ x_{1j} = c }} 1 \right)
    \end{align*}
    Note that
    $$
        \sum_{\substack{j \ne i \\ x_{1j} = c }} 1 
        = |\{ j : j \ne i, x_{1j} = c \}|
        = n_c - 1\{x_{1i} = c\},
    $$
    which gives the desired result.
\end{proof}

\begin{proof}[Proof of Proposition \ref{prop:conservative-bounds}]
    \hypertarget{pf:conservative-bounds}{}
    Recall from definition \ref{def:estimator-for-r} that
    $$\hat{r} = \frac{1}{n(n-1)}\sum_{i = 1}^n \sum_{j \ne i} L( y_i, \hat{f}(x_{1j}, x_{2i}) ).$$
    For each $i \ne j$, denote the shuffled pair of covariates $x = (x_{1j}, x_{2i})$. 
    Since $y_i \in \{0, 1\}$ and $0 \leq \hat{f}_{min}(x) \leq \hat{f}(x) \leq \hat{f}_{max}(x) \leq 1$, either $\hat{f}_{min}(x)$ or $\hat{f}_{max}(x)$ is farther from $y_i$ than $\hat{f}(x)$ is. Therefore,
    $$L(y_i, \hat{f}(x)) \leq L_{max} \equiv \max\{ L(y_i, \hat{f}_{min}(x)), L(y_i, \hat{f}_{max}(x)) \}$$ 
    because 
    $L$ increases monotonically with $|y - \hat{y}|$.
    
    To compute the lower bound $L_{min}$, note either $y_i \in [\hat{f}_{min}(x), \hat{f}_{max}(x)]$ or $y_i$ is outside of this range.
    If $y_i$ is outside of the range, then by the same logic as above, $L(y_i, \hat{f}(x)) \geq \min\{ L(y_i, \hat{f}_{min}(x)), L(y_i, \hat{f}_{max}(x)) \}$.
    If, however, $y_i$ is within this range, then $\hat{f}_{min}(x) = 0$ or $\hat{f}_{max}(x) = 1$ since $y_i \in \{0,1\}$ and $\mathcal{Y} = [0,1]$. Therefore, $\min\{ L(\hat{f}_{min}(x)), L(\hat{f}_{max}(x)) \} = 0$. Thus, either way,
    $$L(y_i, \hat{f}(x)) \geq L_{min} \equiv \max\{ L(y_i, \hat{f}_{min}(x)), L(y_i, \hat{f}_{max}(x)) \}$$ 
    Hence, $L(y_i, \hat{f}(x)) \in [L_{min}, L_{max}]$, so the desired result follows directly.
\end{proof}

\begin{proof}[Proof of Proposition \ref{prop:school-admissions-tight-identification}]
    \hypertarget{pf:school-admissions-tight-identification}{}
    To show that the lower bound is tight, consider the model
    $$Y = 1\{ X_1 + X_2 \geq 1 \}, \, Z = X_2.$$
    That is, $Y = 1\{X_1 + Z \geq 1 \}$. Then, by the law of iterated expectation,
    \begin{align*}
        E[Y \mid X_1 = 0] 
        &= E_Z[E[Y \mid X_1 = 0, Z]] \\
        &= E[Y \mid X_1 = 0, Z = 1]P(Z = 1) + E[Y \mid X_1 = 0, Z = 0]P(Z = 0) \\
        &= E[Y \mid X_1 = 0, Z = 1]P(Z = 1)
    \end{align*}
    since $E[Y \mid X = 0, Z = 0] = 0$. Thus, the lower bound is tight.
    
    Similarly, to see that the upper bound is tight, consider the model
    $$Y = 1\{ X_1 - X_2 \geq 0 \}, \, Z = X_2.$$
    That is, $Y = 1\{ X_1 - Z \geq 0 \}$, and
    \begin{align*}
        E[Y \mid X_1 = 0] 
        &= E_Z[E[Y \mid X_1 = 0, Z]] \\
        &= E[Y \mid X_1 = 0, Z = 1]P(Z = 1) + E[Y \mid X_1 = 0, Z = 0]P(Z = 0) \\
        &= E[Y \mid X_1 = 0, Z = 1]P(Z = 1) + P(Z = 0)
    \end{align*}
    since $E[Y \mid X_1 = 0, Z = 0] = 1$. Thus, both bounds are tight.
\end{proof}

\begin{table}
\centerline{
\begin{tabular}{lrrrlrrr}
\\
\toprule
{} & \multicolumn{4}{l}{Reliance} & \multicolumn{3}{l}{Huber Estimate Coefficient} \\
{} &    Gender & Experience & \multicolumn{2}{l}{Alignment} &    Gender & Experience &  Alignment \\
Justice          &           &            &           &    &           &            &           \\
\midrule
Clarence Thomas     &  6.194115 &   0.914656 &  1.074321 &    & -4.066875 &   0.045477 & -0.967887 \\
William Brennan     &  2.024609 &  -0.012736 &  0.830177 &    &  2.069208 &   0.007237 &  0.953094 \\
John Roberts        &  1.028896 &   0.124092 & -0.179966 &    &  1.231796 &   0.013199 & -0.665908 \\
Antonin Scalia      &  0.720167 &   0.214415 & -0.038943 &    &  1.240160 &   0.010162 & -0.116746 \\
Lewis Powell        &  0.673571 &   0.242610 &  0.019986 &    & -1.705089 &  -0.015079 & -0.063667 \\
Sandra Day O'Connor &  0.543206 &   0.090279 &  0.133488 &    &  1.336867 &   0.011487 & -0.851091 \\
Anthony Kennedy     &  0.521350 &   0.101282 &  0.437427 &    &  1.079679 &  -0.003048 & -0.948730 \\
Stephen Breyer      &  0.452215 &   0.155556 &  0.153670 &    &  1.030942 &  -0.004663 & -0.779379 \\
David Souter        &  0.337769 &  -0.008472 &  0.024470 &    &  1.030264 &   0.001045 & -0.996237 \\
William Rehnquist   &  0.244351 &   0.105431 &  0.445602 &    &  0.737906 &   0.012149 & -0.908850 \\
John Stevens        &  0.232185 &   0.052248 &  0.900550 &    &  0.799289 &   0.012740 & -1.347174 \\
Brett Kavanaugh     &  0.191025 &   1.165301 &  0.924661 &    & -0.339167 &  -0.014464 & -0.696940 \\
Harry Blackmun      &  0.182901 &   0.130280 & -0.002453 &    &  0.829614 &  -0.009206 & -0.474451 \\
Ruth Ginsburg       &  0.143509 &   0.114229 & -0.219453 &    &  0.522730 &  -0.004648 & -0.834549 \\
Elena Kagan         &  0.129975 &   0.209616 & -0.342419 &    &  0.159647 &  -0.005856 & -1.090435 \\
Samuel Alito        &  0.036339 &   0.045138 &  0.279884 &    &  0.397898 &  -0.003541 & -0.768242 \\
Neil Gorsuch        &  0.005193 &   0.118988 &  0.425556 &    &  0.058726 &   0.015541 & -0.688322 \\
Byron White         &  0.003446 &   0.028180 &  0.234361 &    &  0.210519 &   0.006359 & -0.573768 \\
Warren Burger       &  0.001972 &   0.003847 &  0.274191 &    &  0.101611 &   0.001788 & -0.712125 \\
Sonia Sotomayor     & -0.001649 &   0.023033 &  1.127161 &    &  0.032598 &  -0.004273 & -1.648903 \\
Thurgood Marshall   & -0.050600 &   0.772543 & -0.121786 &    &  0.235322 &  -0.028233 & -0.423940 \\
\bottomrule
\\
\end{tabular}
}
\caption{Reliance and Coefficients}
\label{table:reliance-and-coefficients}
\end{table}

\section*{Acknowledgments}

This research was supported by funding from the Cornell University College of Arts and Sciences.

\bibliography{refs}

\end{document}

%% file: pdftex/r_admit_svg-tex.pdf_tex
\begingroup%
  \makeatletter%
  \providecommand\color[2][]{%
    \errmessage{(Inkscape) Color is used for the text in Inkscape, but the package 'color.sty' is not loaded}%
    \renewcommand\color[2][]{}%
  }%
  \providecommand\transparent[1]{%
    \errmessage{(Inkscape) Transparency is used (non-zero) for the text in Inkscape, but the package 'transparent.sty' is not loaded}%
    \renewcommand\transparent[1]{}%
  }%
  \providecommand\rotatebox[2]{#2}%
  \newcommand*\fsize{\dimexpr\f@size pt\relax}%
  \newcommand*\lineheight[1]{\fontsize{\fsize}{#1\fsize}\selectfont}%
  \ifx\svgwidth\undefined%
    \setlength{\unitlength}{392.14375bp}%
    \ifx\svgscale\undefined%
      \relax%
    \else%
      \setlength{\unitlength}{\unitlength * \real{\svgscale}}%
    \fi%
  \else%
    \setlength{\unitlength}{\svgwidth}%
  \fi%
  \global\let\svgwidth\undefined%
  \global\let\svgscale\undefined%
  \makeatother%
  \begin{picture}(1,0.63445165)%
    \lineheight{1}%
    \setlength\tabcolsep{0pt}%
    \put(0,0){\includegraphics[width=\unitlength,page=1]{r_admit_svg-tex.pdf}}%
    \put(0.24429382,0.02437325){\makebox(0,0)[t]{\lineheight{1.25}\smash{\begin{tabular}[t]{c}Race ($X_1$)\end{tabular}}}}%
    \put(0,0){\includegraphics[width=\unitlength,page=2]{r_admit_svg-tex.pdf}}%
    \put(0.5547551,0.02437325){\makebox(0,0)[t]{\lineheight{1.25}\smash{\begin{tabular}[t]{c}Sex ($X_2$)\end{tabular}}}}%
    \put(0,0){\includegraphics[width=\unitlength,page=3]{r_admit_svg-tex.pdf}}%
    \put(0.86521643,0.02437325){\makebox(0,0)[t]{\lineheight{1.25}\smash{\begin{tabular}[t]{c}Score ($X_3$)\end{tabular}}}}%
    \put(0,0){\includegraphics[width=\unitlength,page=4]{r_admit_svg-tex.pdf}}%
    \put(0.11002024,0.05191215){\makebox(0,0)[rt]{\lineheight{1.25}\smash{\begin{tabular}[t]{r}0.00\end{tabular}}}}%
    \put(0,0){\includegraphics[width=\unitlength,page=5]{r_admit_svg-tex.pdf}}%
    \put(0.11002024,0.16743102){\makebox(0,0)[rt]{\lineheight{1.25}\smash{\begin{tabular}[t]{r}0.05\end{tabular}}}}%
    \put(0,0){\includegraphics[width=\unitlength,page=6]{r_admit_svg-tex.pdf}}%
    \put(0.11002024,0.28294989){\makebox(0,0)[rt]{\lineheight{1.25}\smash{\begin{tabular}[t]{r}0.10\end{tabular}}}}%
    \put(0,0){\includegraphics[width=\unitlength,page=7]{r_admit_svg-tex.pdf}}%
    \put(0.11002024,0.39846874){\makebox(0,0)[rt]{\lineheight{1.25}\smash{\begin{tabular}[t]{r}0.15\end{tabular}}}}%
    \put(0,0){\includegraphics[width=\unitlength,page=8]{r_admit_svg-tex.pdf}}%
    \put(0.11002024,0.51398761){\makebox(0,0)[rt]{\lineheight{1.25}\smash{\begin{tabular}[t]{r}0.20\end{tabular}}}}%
    \put(0.03773727,0.33884577){\rotatebox{90}{\makebox(0,0)[t]{\lineheight{1.25}\smash{\begin{tabular}[t]{c}Reliance\end{tabular}}}}}%
    \put(0,0){\includegraphics[width=\unitlength,page=9]{r_admit_svg-tex.pdf}}%
  \end{picture}%
\endgroup%

%% file: pdftex/interruptions_per_chunks_svg-tex.pdf_tex
\begingroup%
  \makeatletter%
  \providecommand\color[2][]{%
    \errmessage{(Inkscape) Color is used for the text in Inkscape, but the package 'color.sty' is not loaded}%
    \renewcommand\color[2][]{}%
  }%
  \providecommand\transparent[1]{%
    \errmessage{(Inkscape) Transparency is used (non-zero) for the text in Inkscape, but the package 'transparent.sty' is not loaded}%
    \renewcommand\transparent[1]{}%
  }%
  \providecommand\rotatebox[2]{#2}%
  \newcommand*\fsize{\dimexpr\f@size pt\relax}%
  \newcommand*\lineheight[1]{\fontsize{\fsize}{#1\fsize}\selectfont}%
  \ifx\svgwidth\undefined%
    \setlength{\unitlength}{605.803125bp}%
    \ifx\svgscale\undefined%
      \relax%
    \else%
      \setlength{\unitlength}{\unitlength * \real{\svgscale}}%
    \fi%
  \else%
    \setlength{\unitlength}{\svgwidth}%
  \fi%
  \global\let\svgwidth\undefined%
  \global\let\svgscale\undefined%
  \makeatother%
  \begin{picture}(1,0.55588552)%
    \lineheight{1}%
    \setlength\tabcolsep{0pt}%
    \put(0,0){\includegraphics[width=\unitlength,page=1]{interruptions_per_chunks_svg-tex.pdf}}%
    \put(0.08570717,0.04628671){\makebox(0,0)[t]{\lineheight{1.25}\smash{\begin{tabular}[t]{c}0\end{tabular}}}}%
    \put(0,0){\includegraphics[width=\unitlength,page=2]{interruptions_per_chunks_svg-tex.pdf}}%
    \put(0.15516268,0.04628671){\makebox(0,0)[t]{\lineheight{1.25}\smash{\begin{tabular}[t]{c}200\end{tabular}}}}%
    \put(0,0){\includegraphics[width=\unitlength,page=3]{interruptions_per_chunks_svg-tex.pdf}}%
    \put(0.22461822,0.04628671){\makebox(0,0)[t]{\lineheight{1.25}\smash{\begin{tabular}[t]{c}400\end{tabular}}}}%
    \put(0,0){\includegraphics[width=\unitlength,page=4]{interruptions_per_chunks_svg-tex.pdf}}%
    \put(0.29407373,0.04628671){\makebox(0,0)[t]{\lineheight{1.25}\smash{\begin{tabular}[t]{c}600\end{tabular}}}}%
    \put(0,0){\includegraphics[width=\unitlength,page=5]{interruptions_per_chunks_svg-tex.pdf}}%
    \put(0.36352925,0.04628671){\makebox(0,0)[t]{\lineheight{1.25}\smash{\begin{tabular}[t]{c}800\end{tabular}}}}%
    \put(0,0){\includegraphics[width=\unitlength,page=6]{interruptions_per_chunks_svg-tex.pdf}}%
    \put(0.43298477,0.04628671){\makebox(0,0)[t]{\lineheight{1.25}\smash{\begin{tabular}[t]{c}1000\end{tabular}}}}%
    \put(0,0){\includegraphics[width=\unitlength,page=7]{interruptions_per_chunks_svg-tex.pdf}}%
    \put(0.05546872,0.08450663){\makebox(0,0)[rt]{\lineheight{1.25}\smash{\begin{tabular}[t]{r}0\end{tabular}}}}%
    \put(0,0){\includegraphics[width=\unitlength,page=8]{interruptions_per_chunks_svg-tex.pdf}}%
    \put(0.05546872,0.15343543){\makebox(0,0)[rt]{\lineheight{1.25}\smash{\begin{tabular}[t]{r}10\end{tabular}}}}%
    \put(0,0){\includegraphics[width=\unitlength,page=9]{interruptions_per_chunks_svg-tex.pdf}}%
    \put(0.05546872,0.22236422){\makebox(0,0)[rt]{\lineheight{1.25}\smash{\begin{tabular}[t]{r}20\end{tabular}}}}%
    \put(0,0){\includegraphics[width=\unitlength,page=10]{interruptions_per_chunks_svg-tex.pdf}}%
    \put(0.05546872,0.29129301){\makebox(0,0)[rt]{\lineheight{1.25}\smash{\begin{tabular}[t]{r}30\end{tabular}}}}%
    \put(0,0){\includegraphics[width=\unitlength,page=11]{interruptions_per_chunks_svg-tex.pdf}}%
    \put(0.05546872,0.36022181){\makebox(0,0)[rt]{\lineheight{1.25}\smash{\begin{tabular}[t]{r}40\end{tabular}}}}%
    \put(0,0){\includegraphics[width=\unitlength,page=12]{interruptions_per_chunks_svg-tex.pdf}}%
    \put(0.05546872,0.42915059){\makebox(0,0)[rt]{\lineheight{1.25}\smash{\begin{tabular}[t]{r}50\end{tabular}}}}%
    \put(0,0){\includegraphics[width=\unitlength,page=13]{interruptions_per_chunks_svg-tex.pdf}}%
    \put(0.05546872,0.49807939){\makebox(0,0)[rt]{\lineheight{1.25}\smash{\begin{tabular}[t]{r}60\end{tabular}}}}%
    \put(0.0244278,0.29471466){\rotatebox{90}{\makebox(0,0)[t]{\lineheight{1.25}\smash{\begin{tabular}[t]{c}Advocate Mean Interruption Rate\end{tabular}}}}}%
    \put(0,0){\includegraphics[width=\unitlength,page=14]{interruptions_per_chunks_svg-tex.pdf}}%
    \put(0.27636257,0.52894917){\makebox(0,0)[t]{\lineheight{1.25}\smash{\begin{tabular}[t]{c}Male Advocates\end{tabular}}}}%
    \put(0,0){\includegraphics[width=\unitlength,page=15]{interruptions_per_chunks_svg-tex.pdf}}%
    \put(0.58692692,0.04628671){\makebox(0,0)[t]{\lineheight{1.25}\smash{\begin{tabular}[t]{c}0\end{tabular}}}}%
    \put(0,0){\includegraphics[width=\unitlength,page=16]{interruptions_per_chunks_svg-tex.pdf}}%
    \put(0.66397475,0.04628671){\makebox(0,0)[t]{\lineheight{1.25}\smash{\begin{tabular}[t]{c}50\end{tabular}}}}%
    \put(0,0){\includegraphics[width=\unitlength,page=17]{interruptions_per_chunks_svg-tex.pdf}}%
    \put(0.74102257,0.04628671){\makebox(0,0)[t]{\lineheight{1.25}\smash{\begin{tabular}[t]{c}100\end{tabular}}}}%
    \put(0,0){\includegraphics[width=\unitlength,page=18]{interruptions_per_chunks_svg-tex.pdf}}%
    \put(0.81807039,0.04628671){\makebox(0,0)[t]{\lineheight{1.25}\smash{\begin{tabular}[t]{c}150\end{tabular}}}}%
    \put(0,0){\includegraphics[width=\unitlength,page=19]{interruptions_per_chunks_svg-tex.pdf}}%
    \put(0.89511826,0.04628671){\makebox(0,0)[t]{\lineheight{1.25}\smash{\begin{tabular}[t]{c}200\end{tabular}}}}%
    \put(0,0){\includegraphics[width=\unitlength,page=20]{interruptions_per_chunks_svg-tex.pdf}}%
    \put(0.97216604,0.04628671){\makebox(0,0)[t]{\lineheight{1.25}\smash{\begin{tabular}[t]{c}250\end{tabular}}}}%
    \put(0,0){\includegraphics[width=\unitlength,page=21]{interruptions_per_chunks_svg-tex.pdf}}%
    \put(0.55788218,0.08450663){\makebox(0,0)[rt]{\lineheight{1.25}\smash{\begin{tabular}[t]{r}0\end{tabular}}}}%
    \put(0,0){\includegraphics[width=\unitlength,page=22]{interruptions_per_chunks_svg-tex.pdf}}%
    \put(0.55788218,0.1346156){\makebox(0,0)[rt]{\lineheight{1.25}\smash{\begin{tabular}[t]{r}10\end{tabular}}}}%
    \put(0,0){\includegraphics[width=\unitlength,page=23]{interruptions_per_chunks_svg-tex.pdf}}%
    \put(0.55788218,0.18472457){\makebox(0,0)[rt]{\lineheight{1.25}\smash{\begin{tabular}[t]{r}20\end{tabular}}}}%
    \put(0,0){\includegraphics[width=\unitlength,page=24]{interruptions_per_chunks_svg-tex.pdf}}%
    \put(0.55788218,0.23483354){\makebox(0,0)[rt]{\lineheight{1.25}\smash{\begin{tabular}[t]{r}30\end{tabular}}}}%
    \put(0,0){\includegraphics[width=\unitlength,page=25]{interruptions_per_chunks_svg-tex.pdf}}%
    \put(0.55788218,0.28494251){\makebox(0,0)[rt]{\lineheight{1.25}\smash{\begin{tabular}[t]{r}40\end{tabular}}}}%
    \put(0,0){\includegraphics[width=\unitlength,page=26]{interruptions_per_chunks_svg-tex.pdf}}%
    \put(0.55788218,0.33505146){\makebox(0,0)[rt]{\lineheight{1.25}\smash{\begin{tabular}[t]{r}50\end{tabular}}}}%
    \put(0,0){\includegraphics[width=\unitlength,page=27]{interruptions_per_chunks_svg-tex.pdf}}%
    \put(0.55788218,0.38516043){\makebox(0,0)[rt]{\lineheight{1.25}\smash{\begin{tabular}[t]{r}60\end{tabular}}}}%
    \put(0,0){\includegraphics[width=\unitlength,page=28]{interruptions_per_chunks_svg-tex.pdf}}%
    \put(0.55788218,0.4352694){\makebox(0,0)[rt]{\lineheight{1.25}\smash{\begin{tabular}[t]{r}70\end{tabular}}}}%
    \put(0,0){\includegraphics[width=\unitlength,page=29]{interruptions_per_chunks_svg-tex.pdf}}%
    \put(0.55788218,0.48537836){\makebox(0,0)[rt]{\lineheight{1.25}\smash{\begin{tabular}[t]{r}80\end{tabular}}}}%
    \put(0,0){\includegraphics[width=\unitlength,page=30]{interruptions_per_chunks_svg-tex.pdf}}%
    \put(0.77877601,0.52894917){\makebox(0,0)[t]{\lineheight{1.25}\smash{\begin{tabular}[t]{c}Female Advocates\end{tabular}}}}%
    \put(0.51271297,0.01531798){\makebox(0,0)[t]{\lineheight{1.25}\smash{\begin{tabular}[t]{c}Number of Chunks with Advocate\end{tabular}}}}%
  \end{picture}%
\endgroup%

%% file: pdftex/ate_svg-tex.pdf_tex
\begingroup%
  \makeatletter%
  \providecommand\color[2][]{%
    \errmessage{(Inkscape) Color is used for the text in Inkscape, but the package 'color.sty' is not loaded}%
    \renewcommand\color[2][]{}%
  }%
  \providecommand\transparent[1]{%
    \errmessage{(Inkscape) Transparency is used (non-zero) for the text in Inkscape, but the package 'transparent.sty' is not loaded}%
    \renewcommand\transparent[1]{}%
  }%
  \providecommand\rotatebox[2]{#2}%
  \newcommand*\fsize{\dimexpr\f@size pt\relax}%
  \newcommand*\lineheight[1]{\fontsize{\fsize}{#1\fsize}\selectfont}%
  \ifx\svgwidth\undefined%
    \setlength{\unitlength}{687.53125bp}%
    \ifx\svgscale\undefined%
      \relax%
    \else%
      \setlength{\unitlength}{\unitlength * \real{\svgscale}}%
    \fi%
  \else%
    \setlength{\unitlength}{\svgwidth}%
  \fi%
  \global\let\svgwidth\undefined%
  \global\let\svgscale\undefined%
  \makeatother%
  \begin{picture}(1,0.69762102)%
    \lineheight{1}%
    \setlength\tabcolsep{0pt}%
    \put(0,0){\includegraphics[width=\unitlength,page=1]{ate_svg-tex.pdf}}%
    \put(0.18321878,0.03339167){\makebox(0,0)[t]{\lineheight{1.25}\smash{\begin{tabular}[t]{c}--15\end{tabular}}}}%
    \put(0,0){\includegraphics[width=\unitlength,page=2]{ate_svg-tex.pdf}}%
    \put(0.33826494,0.03339167){\makebox(0,0)[t]{\lineheight{1.25}\smash{\begin{tabular}[t]{c}--10\end{tabular}}}}%
    \put(0,0){\includegraphics[width=\unitlength,page=3]{ate_svg-tex.pdf}}%
    \put(0.49331107,0.03339167){\makebox(0,0)[t]{\lineheight{1.25}\smash{\begin{tabular}[t]{c}--5\end{tabular}}}}%
    \put(0,0){\includegraphics[width=\unitlength,page=4]{ate_svg-tex.pdf}}%
    \put(0.64835725,0.03339167){\makebox(0,0)[t]{\lineheight{1.25}\smash{\begin{tabular}[t]{c}0\end{tabular}}}}%
    \put(0,0){\includegraphics[width=\unitlength,page=5]{ate_svg-tex.pdf}}%
    \put(0.80340338,0.03339167){\makebox(0,0)[t]{\lineheight{1.25}\smash{\begin{tabular}[t]{c}5\end{tabular}}}}%
    \put(0,0){\includegraphics[width=\unitlength,page=6]{ate_svg-tex.pdf}}%
    \put(0.9584496,0.03339167){\makebox(0,0)[t]{\lineheight{1.25}\smash{\begin{tabular}[t]{c}10\end{tabular}}}}%
    \put(0.58372799,0.01349711){\makebox(0,0)[t]{\lineheight{1.25}\smash{\begin{tabular}[t]{c}Interruptions per 1,000 Tokens\end{tabular}}}}%
    \put(0,0){\includegraphics[width=\unitlength,page=7]{ate_svg-tex.pdf}}%
    \put(0.16774693,0.65287179){\makebox(0,0)[rt]{\lineheight{1.25}\smash{\begin{tabular}[t]{r}Clarence Thomas\end{tabular}}}}%
    \put(0,0){\includegraphics[width=\unitlength,page=8]{ate_svg-tex.pdf}}%
    \put(0.16774693,0.6241207){\makebox(0,0)[rt]{\lineheight{1.25}\smash{\begin{tabular}[t]{r}Lewis Powell\end{tabular}}}}%
    \put(0,0){\includegraphics[width=\unitlength,page=9]{ate_svg-tex.pdf}}%
    \put(0.16774693,0.59536961){\makebox(0,0)[rt]{\lineheight{1.25}\smash{\begin{tabular}[t]{r}Brett Kavanaugh\end{tabular}}}}%
    \put(0,0){\includegraphics[width=\unitlength,page=10]{ate_svg-tex.pdf}}%
    \put(0.16774693,0.56661852){\makebox(0,0)[rt]{\lineheight{1.25}\smash{\begin{tabular}[t]{r}Sonia Sotomayor\end{tabular}}}}%
    \put(0,0){\includegraphics[width=\unitlength,page=11]{ate_svg-tex.pdf}}%
    \put(0.16774693,0.53786744){\makebox(0,0)[rt]{\lineheight{1.25}\smash{\begin{tabular}[t]{r}Neil Gorsuch\end{tabular}}}}%
    \put(0,0){\includegraphics[width=\unitlength,page=12]{ate_svg-tex.pdf}}%
    \put(0.16774693,0.50911634){\makebox(0,0)[rt]{\lineheight{1.25}\smash{\begin{tabular}[t]{r}Warren Burger\end{tabular}}}}%
    \put(0,0){\includegraphics[width=\unitlength,page=13]{ate_svg-tex.pdf}}%
    \put(0.16774693,0.48036525){\makebox(0,0)[rt]{\lineheight{1.25}\smash{\begin{tabular}[t]{r}Elena Kagan\end{tabular}}}}%
    \put(0,0){\includegraphics[width=\unitlength,page=14]{ate_svg-tex.pdf}}%
    \put(0.16774693,0.45161416){\makebox(0,0)[rt]{\lineheight{1.25}\smash{\begin{tabular}[t]{r}Byron White\end{tabular}}}}%
    \put(0,0){\includegraphics[width=\unitlength,page=15]{ate_svg-tex.pdf}}%
    \put(0.16774693,0.42286306){\makebox(0,0)[rt]{\lineheight{1.25}\smash{\begin{tabular}[t]{r}Thurgood Marshall\end{tabular}}}}%
    \put(0,0){\includegraphics[width=\unitlength,page=16]{ate_svg-tex.pdf}}%
    \put(0.16774693,0.39411197){\makebox(0,0)[rt]{\lineheight{1.25}\smash{\begin{tabular}[t]{r}Samuel Alito\end{tabular}}}}%
    \put(0,0){\includegraphics[width=\unitlength,page=17]{ate_svg-tex.pdf}}%
    \put(0.16774693,0.36536088){\makebox(0,0)[rt]{\lineheight{1.25}\smash{\begin{tabular}[t]{r}Ruth Ginsburg\end{tabular}}}}%
    \put(0,0){\includegraphics[width=\unitlength,page=18]{ate_svg-tex.pdf}}%
    \put(0.16774693,0.33660979){\makebox(0,0)[rt]{\lineheight{1.25}\smash{\begin{tabular}[t]{r}William Rehnquist\end{tabular}}}}%
    \put(0,0){\includegraphics[width=\unitlength,page=19]{ate_svg-tex.pdf}}%
    \put(0.16774693,0.30785873){\makebox(0,0)[rt]{\lineheight{1.25}\smash{\begin{tabular}[t]{r}John Stevens\end{tabular}}}}%
    \put(0,0){\includegraphics[width=\unitlength,page=20]{ate_svg-tex.pdf}}%
    \put(0.16774693,0.27910764){\makebox(0,0)[rt]{\lineheight{1.25}\smash{\begin{tabular}[t]{r}Harry Blackmun\end{tabular}}}}%
    \put(0,0){\includegraphics[width=\unitlength,page=21]{ate_svg-tex.pdf}}%
    \put(0.16774693,0.25035655){\makebox(0,0)[rt]{\lineheight{1.25}\smash{\begin{tabular}[t]{r}David Souter\end{tabular}}}}%
    \put(0,0){\includegraphics[width=\unitlength,page=22]{ate_svg-tex.pdf}}%
    \put(0.16774693,0.22160546){\makebox(0,0)[rt]{\lineheight{1.25}\smash{\begin{tabular}[t]{r}Stephen Breyer\end{tabular}}}}%
    \put(0,0){\includegraphics[width=\unitlength,page=23]{ate_svg-tex.pdf}}%
    \put(0.16774693,0.19285437){\makebox(0,0)[rt]{\lineheight{1.25}\smash{\begin{tabular}[t]{r}Anthony Kennedy\end{tabular}}}}%
    \put(0,0){\includegraphics[width=\unitlength,page=24]{ate_svg-tex.pdf}}%
    \put(0.16774693,0.16410328){\makebox(0,0)[rt]{\lineheight{1.25}\smash{\begin{tabular}[t]{r}John Roberts\end{tabular}}}}%
    \put(0,0){\includegraphics[width=\unitlength,page=25]{ate_svg-tex.pdf}}%
    \put(0.16774693,0.13535219){\makebox(0,0)[rt]{\lineheight{1.25}\smash{\begin{tabular}[t]{r}Antonin Scalia\end{tabular}}}}%
    \put(0,0){\includegraphics[width=\unitlength,page=26]{ate_svg-tex.pdf}}%
    \put(0.16774693,0.1066011){\makebox(0,0)[rt]{\lineheight{1.25}\smash{\begin{tabular}[t]{r}Sandra Day O'Connor\end{tabular}}}}%
    \put(0,0){\includegraphics[width=\unitlength,page=27]{ate_svg-tex.pdf}}%
    \put(0.16774693,0.07785001){\makebox(0,0)[rt]{\lineheight{1.25}\smash{\begin{tabular}[t]{r}William Brennan\end{tabular}}}}%
    \put(0,0){\includegraphics[width=\unitlength,page=28]{ate_svg-tex.pdf}}%
  \end{picture}%
\endgroup%

%% file: pdftex/reliance_svg-tex.pdf_tex
\begingroup%
  \makeatletter%
  \providecommand\color[2][]{%
    \errmessage{(Inkscape) Color is used for the text in Inkscape, but the package 'color.sty' is not loaded}%
    \renewcommand\color[2][]{}%
  }%
  \providecommand\transparent[1]{%
    \errmessage{(Inkscape) Transparency is used (non-zero) for the text in Inkscape, but the package 'transparent.sty' is not loaded}%
    \renewcommand\transparent[1]{}%
  }%
  \providecommand\rotatebox[2]{#2}%
  \newcommand*\fsize{\dimexpr\f@size pt\relax}%
  \newcommand*\lineheight[1]{\fontsize{\fsize}{#1\fsize}\selectfont}%
  \ifx\svgwidth\undefined%
    \setlength{\unitlength}{687.53125bp}%
    \ifx\svgscale\undefined%
      \relax%
    \else%
      \setlength{\unitlength}{\unitlength * \real{\svgscale}}%
    \fi%
  \else%
    \setlength{\unitlength}{\svgwidth}%
  \fi%
  \global\let\svgwidth\undefined%
  \global\let\svgscale\undefined%
  \makeatother%
  \begin{picture}(1,0.69762102)%
    \lineheight{1}%
    \setlength\tabcolsep{0pt}%
    \put(0,0){\includegraphics[width=\unitlength,page=1]{reliance_svg-tex.pdf}}%
    \put(0.23532568,0.03339167){\makebox(0,0)[t]{\lineheight{1.25}\smash{\begin{tabular}[t]{c}0\end{tabular}}}}%
    \put(0,0){\includegraphics[width=\unitlength,page=2]{reliance_svg-tex.pdf}}%
    \put(0.35012052,0.03339167){\makebox(0,0)[t]{\lineheight{1.25}\smash{\begin{tabular}[t]{c}1\end{tabular}}}}%
    \put(0,0){\includegraphics[width=\unitlength,page=3]{reliance_svg-tex.pdf}}%
    \put(0.46491535,0.03339167){\makebox(0,0)[t]{\lineheight{1.25}\smash{\begin{tabular}[t]{c}2\end{tabular}}}}%
    \put(0,0){\includegraphics[width=\unitlength,page=4]{reliance_svg-tex.pdf}}%
    \put(0.5797102,0.03339167){\makebox(0,0)[t]{\lineheight{1.25}\smash{\begin{tabular}[t]{c}3\end{tabular}}}}%
    \put(0,0){\includegraphics[width=\unitlength,page=5]{reliance_svg-tex.pdf}}%
    \put(0.69450501,0.03339167){\makebox(0,0)[t]{\lineheight{1.25}\smash{\begin{tabular}[t]{c}4\end{tabular}}}}%
    \put(0,0){\includegraphics[width=\unitlength,page=6]{reliance_svg-tex.pdf}}%
    \put(0.80929986,0.03339167){\makebox(0,0)[t]{\lineheight{1.25}\smash{\begin{tabular}[t]{c}5\end{tabular}}}}%
    \put(0,0){\includegraphics[width=\unitlength,page=7]{reliance_svg-tex.pdf}}%
    \put(0.92409471,0.03339167){\makebox(0,0)[t]{\lineheight{1.25}\smash{\begin{tabular}[t]{c}6\end{tabular}}}}%
    \put(0.58372799,0.01349711){\makebox(0,0)[t]{\lineheight{1.25}\smash{\begin{tabular}[t]{c}Reliance\end{tabular}}}}%
    \put(0,0){\includegraphics[width=\unitlength,page=8]{reliance_svg-tex.pdf}}%
    \put(0.16774693,0.65287179){\makebox(0,0)[rt]{\lineheight{1.25}\smash{\begin{tabular}[t]{r}John Roberts\end{tabular}}}}%
    \put(0,0){\includegraphics[width=\unitlength,page=9]{reliance_svg-tex.pdf}}%
    \put(0.16774693,0.6241207){\makebox(0,0)[rt]{\lineheight{1.25}\smash{\begin{tabular}[t]{r}Antonin Scalia\end{tabular}}}}%
    \put(0,0){\includegraphics[width=\unitlength,page=10]{reliance_svg-tex.pdf}}%
    \put(0.16774693,0.59536961){\makebox(0,0)[rt]{\lineheight{1.25}\smash{\begin{tabular}[t]{r}Lewis Powell\end{tabular}}}}%
    \put(0,0){\includegraphics[width=\unitlength,page=11]{reliance_svg-tex.pdf}}%
    \put(0.16774693,0.56661852){\makebox(0,0)[rt]{\lineheight{1.25}\smash{\begin{tabular}[t]{r}Stephen Breyer\end{tabular}}}}%
    \put(0,0){\includegraphics[width=\unitlength,page=12]{reliance_svg-tex.pdf}}%
    \put(0.16774693,0.53786744){\makebox(0,0)[rt]{\lineheight{1.25}\smash{\begin{tabular}[t]{r}Harry Blackmun\end{tabular}}}}%
    \put(0,0){\includegraphics[width=\unitlength,page=13]{reliance_svg-tex.pdf}}%
    \put(0.16774693,0.50911634){\makebox(0,0)[rt]{\lineheight{1.25}\smash{\begin{tabular}[t]{r}Ruth Ginsburg\end{tabular}}}}%
    \put(0,0){\includegraphics[width=\unitlength,page=14]{reliance_svg-tex.pdf}}%
    \put(0.16774693,0.48036525){\makebox(0,0)[rt]{\lineheight{1.25}\smash{\begin{tabular}[t]{r}Clarence Thomas\end{tabular}}}}%
    \put(0,0){\includegraphics[width=\unitlength,page=15]{reliance_svg-tex.pdf}}%
    \put(0.16774693,0.45161416){\makebox(0,0)[rt]{\lineheight{1.25}\smash{\begin{tabular}[t]{r}William Brennan\end{tabular}}}}%
    \put(0,0){\includegraphics[width=\unitlength,page=16]{reliance_svg-tex.pdf}}%
    \put(0.16774693,0.42286306){\makebox(0,0)[rt]{\lineheight{1.25}\smash{\begin{tabular}[t]{r}Sandra Day O'Connor\end{tabular}}}}%
    \put(0,0){\includegraphics[width=\unitlength,page=17]{reliance_svg-tex.pdf}}%
    \put(0.16774693,0.39411197){\makebox(0,0)[rt]{\lineheight{1.25}\smash{\begin{tabular}[t]{r}Anthony Kennedy\end{tabular}}}}%
    \put(0,0){\includegraphics[width=\unitlength,page=18]{reliance_svg-tex.pdf}}%
    \put(0.16774693,0.36536088){\makebox(0,0)[rt]{\lineheight{1.25}\smash{\begin{tabular}[t]{r}David Souter\end{tabular}}}}%
    \put(0,0){\includegraphics[width=\unitlength,page=19]{reliance_svg-tex.pdf}}%
    \put(0.16774693,0.33660979){\makebox(0,0)[rt]{\lineheight{1.25}\smash{\begin{tabular}[t]{r}Elena Kagan\end{tabular}}}}%
    \put(0,0){\includegraphics[width=\unitlength,page=20]{reliance_svg-tex.pdf}}%
    \put(0.16774693,0.30785873){\makebox(0,0)[rt]{\lineheight{1.25}\smash{\begin{tabular}[t]{r}Thurgood Marshall\end{tabular}}}}%
    \put(0,0){\includegraphics[width=\unitlength,page=21]{reliance_svg-tex.pdf}}%
    \put(0.16774693,0.27910764){\makebox(0,0)[rt]{\lineheight{1.25}\smash{\begin{tabular}[t]{r}Brett Kavanaugh\end{tabular}}}}%
    \put(0,0){\includegraphics[width=\unitlength,page=22]{reliance_svg-tex.pdf}}%
    \put(0.16774693,0.25035655){\makebox(0,0)[rt]{\lineheight{1.25}\smash{\begin{tabular}[t]{r}William Rehnquist\end{tabular}}}}%
    \put(0,0){\includegraphics[width=\unitlength,page=23]{reliance_svg-tex.pdf}}%
    \put(0.16774693,0.22160546){\makebox(0,0)[rt]{\lineheight{1.25}\smash{\begin{tabular}[t]{r}John Stevens\end{tabular}}}}%
    \put(0,0){\includegraphics[width=\unitlength,page=24]{reliance_svg-tex.pdf}}%
    \put(0.16774693,0.19285437){\makebox(0,0)[rt]{\lineheight{1.25}\smash{\begin{tabular}[t]{r}Samuel Alito\end{tabular}}}}%
    \put(0,0){\includegraphics[width=\unitlength,page=25]{reliance_svg-tex.pdf}}%
    \put(0.16774693,0.16410328){\makebox(0,0)[rt]{\lineheight{1.25}\smash{\begin{tabular}[t]{r}Neil Gorsuch\end{tabular}}}}%
    \put(0,0){\includegraphics[width=\unitlength,page=26]{reliance_svg-tex.pdf}}%
    \put(0.16774693,0.13535219){\makebox(0,0)[rt]{\lineheight{1.25}\smash{\begin{tabular}[t]{r}Byron White\end{tabular}}}}%
    \put(0,0){\includegraphics[width=\unitlength,page=27]{reliance_svg-tex.pdf}}%
    \put(0.16774693,0.1066011){\makebox(0,0)[rt]{\lineheight{1.25}\smash{\begin{tabular}[t]{r}Warren Burger\end{tabular}}}}%
    \put(0,0){\includegraphics[width=\unitlength,page=28]{reliance_svg-tex.pdf}}%
    \put(0.16774693,0.07785001){\makebox(0,0)[rt]{\lineheight{1.25}\smash{\begin{tabular}[t]{r}Sonia Sotomayor\end{tabular}}}}%
    \put(0,0){\includegraphics[width=\unitlength,page=29]{reliance_svg-tex.pdf}}%
    \put(0.5797102,0.58076974){\makebox(0,0)[lt]{\lineheight{1.25}\smash{\begin{tabular}[t]{l}Alignment < Experience < Gender\end{tabular}}}}%
    \put(0.5797102,0.42263873){\makebox(0,0)[lt]{\lineheight{1.25}\smash{\begin{tabular}[t]{l}Experience < Alignment < Gender\end{tabular}}}}%
    \put(0.5797102,0.32200994){\makebox(0,0)[lt]{\lineheight{1.25}\smash{\begin{tabular}[t]{l}Alignment < Gender < Experience\end{tabular}}}}%
    \put(0.5797102,0.2788833){\makebox(0,0)[lt]{\lineheight{1.25}\smash{\begin{tabular}[t]{l}Gender < Alignment < Experience\end{tabular}}}}%
    \put(0.5797102,0.23575667){\makebox(0,0)[lt]{\lineheight{1.25}\smash{\begin{tabular}[t]{l}Experience < Gender < Alignment\end{tabular}}}}%
    \put(0.5797102,0.13512785){\makebox(0,0)[lt]{\lineheight{1.25}\smash{\begin{tabular}[t]{l}Gender < Experience < Alignment\end{tabular}}}}%
  \end{picture}%
\endgroup%

%% file: main.bbl
\begin{thebibliography}{23}
\providecommand{\natexlab}[1]{#1}
\providecommand{\url}[1]{\texttt{#1}}
\expandafter\ifx\csname urlstyle\endcsname\relax
  \providecommand{\doi}[1]{doi: #1}\else
  \providecommand{\doi}{doi: \begingroup \urlstyle{rm}\Url}\fi

\bibitem[Goulette et~al.(2015)Goulette, Wooldredge, Frank, and
  Travis]{Goulette}
Natalie Goulette, John Wooldredge, James Frank, and Lawrence Travis.
\newblock From initial appearance to sentencing: Do female defendants
  experience disparate treatment?
\newblock \emph{Journal of Criminal Justice}, 43\penalty0 (5):\penalty0
  406--417, 2015.
\newblock ISSN 0047-2352.
\newblock \doi{https://doi.org/10.1016/j.jcrimjus.2015.07.003}.
\newblock URL
  \url{https://www.sciencedirect.com/science/article/pii/S0047235215000665}.

\bibitem[Staton et~al.(2007)Staton, Panda, Chen, Genao, Kurz, Pasanen,
  Mechaber, Menon, O'Rorke, Wood, Rosenberg, Faeslis, Carey, Calleson, and
  Cykert]{Staton}
Lisa~J. Staton, Mukta Panda, Ian Chen, Inginia Genao, James Kurz, Mark Pasanen,
  Alex~J. Mechaber, Madhusudan Menon, Jane O'Rorke, JoAnn Wood, Eric Rosenberg,
  Charles Faeslis, Tim Carey, Diane Calleson, and Sam Cykert.
\newblock When race matters: disagreement in pain perception between patients
  and their physicians in primary care.
\newblock \emph{Journal of the National Medical Association}, 99\penalty0
  (5):\penalty0 532--538, 2007.

\bibitem[sff(2022)]{sffa}
{Students for Fair Admissions v. President and Fellows of Harvard College},
  2022.
\newblock URL \url{www.oyez.org/cases/2022/20-1199}.

\bibitem[Breiman(2002)]{Breiman_Manual}
Leo Breiman.
\newblock Manual on setting up, using, and understanding random forests v3. 1.
\newblock page~14, 2002.

\bibitem[Fisher et~al.(2019)Fisher, Rudin, and Dominici]{Fisher}
Aaron Fisher, Cynthia Rudin, and Francesca Dominici.
\newblock All models are wrong, but many are useful: Learning a variable’s
  importance by studying an entire class of prediction models simultaneously.
\newblock \emph{Journal of Machine Learning Research}, 20, 12 2019.
\newblock URL \url{http://arxiv.org/abs/1801.01489}.

\bibitem[Coyle et~al.(2003)Coyle, Buxton, and O’Brien]{Coyle}
Douglas Coyle, Martin~J Buxton, and Bernie~J O’Brien.
\newblock Measures of importance for economic analysis based on decision
  modeling.
\newblock \emph{Journal of Clinical Epidemiology}, 56\penalty0 (10):\penalty0
  989–997, 10 2003.
\newblock ISSN 08954356.
\newblock \doi{10.1016/S0895-4356(03)00176-8}.

\bibitem[Barredo~Arrieta et~al.(2020)Barredo~Arrieta, Díaz-Rodríguez,
  Del~Ser, Bennetot, Tabik, Barbado, Garcia, Gil-Lopez, Molina, Benjamins,
  Chatila, and Herrera]{Barredo}
Alejandro Barredo~Arrieta, Natalia Díaz-Rodríguez, Javier Del~Ser, Adrien
  Bennetot, Siham Tabik, Alberto Barbado, Salvador Garcia, Sergio Gil-Lopez,
  Daniel Molina, Richard Benjamins, Raja Chatila, and Francisco Herrera.
\newblock Explainable artificial intelligence (xai): Concepts, taxonomies,
  opportunities and challenges toward responsible ai.
\newblock \emph{Information Fusion}, 58:\penalty0 82–115, 6 2020.
\newblock ISSN 15662535.
\newblock \doi{10.1016/j.inffus.2019.12.012}.

\bibitem[Lundberg and Lee(2017)]{Lundberg}
Scott~M. Lundberg and Su-In Lee.
\newblock A unified approach to interpreting model predictions.
\newblock In \emph{Proceedings of the 31st International Conference on Neural
  Information Processing Systems}, NIPS'17, page 4768–4777, Red Hook, NY,
  USA, 2017. Curran Associates Inc.
\newblock ISBN 9781510860964.

\bibitem[Robnik-{\v{S}}ikonja and Bohanec(2018)]{Robnik}
Marko Robnik-{\v{S}}ikonja and Marko Bohanec.
\newblock \emph{Perturbation-Based Explanations of Prediction Models}, pages
  159--175.
\newblock Springer International Publishing, Cham, 2018.
\newblock ISBN 978-3-319-90403-0.
\newblock \doi{10.1007/978-3-319-90403-0_9}.
\newblock URL \url{https://doi.org/10.1007/978-3-319-90403-0_9}.

\bibitem[Breiman(2001)]{Breiman}
Leo Breiman.
\newblock Random forests.
\newblock \emph{Machine learning}, 45:\penalty0 5--32, 2001.

\bibitem[Balagopalan et~al.(2022)Balagopalan, Zhang, Hamidieh, Hartvigsen,
  Rudzicz, and Ghassemi]{Balagopalan}
Aparna Balagopalan, Haoran Zhang, Kimia Hamidieh, Thomas Hartvigsen, Frank
  Rudzicz, and Marzyeh Ghassemi.
\newblock The road to explainability is paved with bias: Measuring the fairness
  of explanations.
\newblock page 1194–1206. ACM, 6 2022.
\newblock ISBN 978-1-4503-9352-2.
\newblock \doi{10.1145/3531146.3533179}.
\newblock URL \url{https://dl.acm.org/doi/10.1145/3531146.3533179}.

\bibitem[Marco~Tulio et~al.(2016)Marco~Tulio, Singh, and Guestrin]{Ribeiro}
Ribeiro Marco~Tulio, Sameer Singh, and Carlos Guestrin.
\newblock Why should {I} trust you?: Explaining the predictions of any
  classifier.
\newblock In \emph{ACM SIGKDD International Conference on Knowledge Discovery
  and Data Mining}, pages 1135--1144. ACM, 2016.

\bibitem[Gregorutti et~al.(2017)Gregorutti, Michel, and
  Saint-Pierre]{Gregorutti_Michel_Saint-Pierre_2017}
Baptiste Gregorutti, Bertrand Michel, and Philippe Saint-Pierre.
\newblock Correlation and variable importance in random forests.
\newblock \emph{Statistics and Computing}, 27\penalty0 (3):\penalty0 659–678,
  5 2017.
\newblock ISSN 0960-3174, 1573-1375.
\newblock \doi{10.1007/s11222-016-9646-1}.

\bibitem[Corbett{-}Davies et~al.(2017)Corbett{-}Davies, Pierson, Feller, Goel,
  and Huq]{DBLP:journals/corr/Corbett-DaviesP17}
Sam Corbett{-}Davies, Emma Pierson, Avi Feller, Sharad Goel, and Aziz Huq.
\newblock Algorithmic decision making and the cost of fairness.
\newblock \emph{CoRR}, abs/1701.08230, 2017.
\newblock URL \url{http://arxiv.org/abs/1701.08230}.

\bibitem[Chang et~al.(2020)Chang, Chiam, Fu, Wang, Zhang, and
  Danescu{-}Niculescu{-}Mizil]{Chang-ConvoKit}
Jonathan~P. Chang, Caleb Chiam, Liye Fu, Andrew~Z. Wang, Justine Zhang, and
  Cristian Danescu{-}Niculescu{-}Mizil.
\newblock Convokit: {A} toolkit for the analysis of conversations.
\newblock \emph{CoRR}, abs/2005.04246, 2020.
\newblock URL \url{https://arxiv.org/abs/2005.04246}.

\bibitem[Danescu-Niculescu-Mizil et~al.(2012)Danescu-Niculescu-Mizil, Lee,
  Pang, and Kleinberg]{Danescu-Niculescu-Mizil+al:12a}
Cristian Danescu-Niculescu-Mizil, Lillian Lee, Bo~Pang, and Jon Kleinberg.
\newblock Echoes of power: {Language} effects and power differences in social
  interaction.
\newblock In \emph{Proceedings of WWW}, pages 699--708, 2012.

\bibitem[Cai et~al.(2023)Cai, Gupta, Keith, O'Connor, and Rice]{cai_et_al_2023}
Erica Cai, Ankita Gupta, Katherine Keith, Brendan O'Connor, and Douglas~R Rice.
\newblock "let me just interrupt you": Estimating gender effects in supreme
  court oral arguments, 2 2023.
\newblock URL \url{osf.io/preprints/socarxiv/4dngy}.

\bibitem[Pérez(2016)]{gender-guesser}
Israel~Saeta Pérez.
\newblock Gender guesser.
\newblock \url{https://github.com/lead-ratings/gender-guesser}, 2016.

\bibitem[Vebman(2023)]{vebman-scorpus}
Daniel Vebman.
\newblock Scorpus.
\newblock \url{https://github.com/danielVebman/scorpus}, 2023.

\bibitem[Wolfson(2002)]{Wolfson}
Warren~D Wolfson.
\newblock Oral argument: Does it matter?
\newblock \emph{Indiana Law Review}, 35:\penalty0 451--456, 2002.
\newblock URL \url{https://mckinneylaw.iu.edu/ilr/pdf/vol35p451.pdf}.

\bibitem[Duvall(2007)]{Duvall}
Michael Duvall.
\newblock When is oral argument important? a judicial clerk's view of the
  debate.
\newblock \emph{The Journal of Appellate Practice and Process}, 9:\penalty0
  121--130, 2007.
\newblock URL
  \url{https://lawrepository.ualr.edu/appellatepracticeprocess/vol9/iss1/5}.

\bibitem[Coleman(2023)]{Coleman}
Robert Coleman, Jr.
\newblock The vanishing oral argument: Why it matters and what to do about it.
\newblock \emph{American Bar Association Appellate Issues}, 2023.
\newblock URL
  \url{https://www.americanbar.org/groups/judicial/publications/appellate_issues/2020/winter/the-vanishing-oral-argument/}.

\bibitem[Huber(1981)]{huber-robust-rlm}
Peter~J. Huber.
\newblock \emph{Regression}, pages 153--198.
\newblock John Wiley \& Sons, Ltd, 1981.
\newblock ISBN 9780471725251.
\newblock \doi{https://doi.org/10.1002/0471725250.ch7}.
\newblock URL
  \url{https://onlinelibrary.wiley.com/doi/abs/10.1002/0471725250.ch7}.

\end{thebibliography}
